%% file: main.tex
\documentclass{vldb}
\usepackage{graphicx} 
\usepackage{balance}

\usepackage{multirow, rotating, amsmath, wasysym, url}
\usepackage[tight]{subfigure}
\usepackage[retainorgcmds]{IEEEtrantools}
\usepackage[ruled,linesnumbered,vlined]{algorithm2e}
\usepackage{hyperref}
\usepackage[usenames,dvipsnames]{xcolor}

\definecolor{greener}{RGB}{0,0,0}
\definecolor{reder}{RGB}{0,0,0}
\definecolor{bluer}{RGB}{0,0,0}

\usepackage{color}
\input{new-commands}

\newcommand{\presec}{\vspace{-0.0in}}
\newcommand{\postsec}{\vspace{-0.0in}}
\newcommand{\presub}{\vspace{-0.0in}}
\newcommand{\postsub}{\vspace{-0.0in}}

\newcommand{\prefigcaption}{\vspace{-0.05in}}
\newcommand{\postfig}{\vspace{-0.05in}}

\graphicspath{{./DrawingFigures/},{./Experimentalfigures/}}
\newcommand{\tabincell}[2]{\begin{tabular}{@{}#1@{}}#2\end{tabular}}
\newcommand{\Proc}{Proc. }

\title{A Shifting Bloom Filter Framework for Set Queries}

\numberofauthors{8} 

\author{
\alignauthor
\vspace{-0.3in}
Tong Yang\\
       \affaddr{Peking University, China}\\\vspace{-0.05in}
       \email{yangtongemail@gmail.com}
\alignauthor
\vspace{-0.3in}
Alex X. Liu\\
       \affaddr{Michigan State University}\\\vspace{-0.05in}
       \email{alexliu@cse.msu.edu}
\alignauthor
\vspace{-0.3in}
Muhammad Shahzad\\
       \affaddr{North\hspace{-0.005in} Carolina\hspace{-0.005in} State\hspace{-0.005in} University}\\\vspace{-0.05in}
       \email{mshahza@ncsu.edu}
\and  
\alignauthor Yuankun Zhong\\
       \affaddr{Nanjing University, China}\\\vspace{-0.05in}
       \email{kun@smail.nju.edu.cn}
\alignauthor Qiaobin Fu\\
       \affaddr{Boston University, USA}\\\vspace{-0.05in}
       \email{qiaobinf@bu.edu}
\alignauthor Zi Li\\
       \affaddr{Nanjing University, China}\\\vspace{-0.05in}
       \email{zizi13572468@gmail.com}
\and
\alignauthor Gaogang Xie\\
       \affaddr{ICT, CAS, China}\\\vspace{-0.05in}
\alignauthor Xiaoming Li\\
       \affaddr{Peking University, China}\\\vspace{-0.05in}
\and
\vspace{0.1in}
 \textcolor{blue}{This paper will appear in VLDB 2016, see \url{http://www.vldb.org/pvldb/vol9.html}}
}

\begin{document}

\maketitle
\sloppy
\input{abstract}
\input{Introduction}
\input{relatedwork.tex}
\input{membership}

\input{generalization-now}

\input{twosets.tex}

\input{multiplicityNew}

\input{evaluation}

\input{conclusion}
\balance
\bibliographystyle{abbrv}

\input{reference.bbl}
\end{document}

%% file: new-commands.tex
\mathchardef\Gamma="0100 \mathchardef\Delta="0101
\mathchardef\Theta="0102 \mathchardef\Lambda="0103
\mathchardef\Xi="0104 \mathchardef\Pi="0105
\mathchardef\Sigma="0106 \mathchardef\Upsilon="0107
\mathchardef\Phi="0108 \mathchardef\Psi="0109
\mathchardef\Omega="010A

\newcommand{\outline}[1]{}

\usepackage{cite}
\usepackage{xspace}
\usepackage{url}
\usepackage{graphicx}
\usepackage{latexsym}
\usepackage{amssymb}
\usepackage{amsfonts}
\usepackage{psfrag}
\usepackage{wrapfig}
\usepackage{comment}
\usepackage{alltt}
\usepackage{color}

\newtheorem{thm}{Theorem}

\newcommand{\ie}{\emph{i.e.}\xspace}

\newcommand{\etal}{\frenchspacing{}\emph{et al{.}}\xspace}

\newcommand{\Comment}[1]{}

%% file: abstract.tex
\begin{abstract}
Set queries are fundamental operations in computer systems and applications.
This paper addresses the fundamental problem of designing a probabilistic data structure that can quickly process set queries using a small amount of memory.
We propose a Shifting Bloom Filter (ShBF) framework for representing and querying sets.
We demonstrate the effectiveness of ShBF using three types of popular set queries: membership, association, and multiplicity queries.
The key novelty of ShBF is on encoding the auxiliary information of a set element in a location offset.
In contrast, prior BF based set data structures allocate additional memory to store auxiliary information.
To evaluate ShBF in comparison with prior art, we conducted experiments using real-world network traces.
Results show that ShBF significantly advances the state-of-the-art on all three types of set queries.
\end{abstract}

%% file: Introduction.tex
\vspace{-0.08in}
\section{Introduction}
\vspace{-0.07in}
\subsection{Motivations}
Set queries, such as \emph{membership queries}, \emph{association queries}, and \emph{multiplicity queries}, are fundamental operations in computer systems and applications.
\emph{Membership queries check whether an element is a member of a given set}.
Network applications, such as IP lookup, packet classification, and regular expression matching, often involve membership queries.
%
{\color{reder}\emph{Association queries identify which set(s) among a pair of sets contain a given element}.
Network architectures such as distributed servers often use association queries.
For example, when data is stored distributively on two servers and the popular content is replicated over both servers to achieve load balancing, for any incoming query, the gateway needs to identify the server(s) that contain the data corresponding to that query.}
%
%
\emph{Multiplicity queries check how many times an element appears in a multi-set}.
A multi-set allows elements to appear more than once.
Network measurement applications, such as measuring flow sizes, often use multiplicity queries.

This paper addresses the fundamental problem of designing a probabilistic data structure that can quickly process set queries, such as the above-mentioned membership, association, and multiplicity queries, using a small amount of memory.
Set query processing speed is critical for many systems and applications, especially for networking applications as packets need to be processed at wire speed.
Memory consumption is also critical because small memory consumption may allow the data structure to be stored in SRAM, which is an order of magnitude faster than DRAM.

Widely used set data structures are the standard Bloom Filter (BF) \cite{BF1970} and the counting Bloom Filter (CBF) \cite{webcaching}.
Let $h_1(.),\cdots,h_k(.)$ be $k$ independent hash functions with uniformly distributed outputs.
Given a set $S$, BF constructs an array $B$ of $m$ bits, where each bit is initialized to 0, and for each element $e$$\in$$S$, BF sets the $k$ bits $B[h_1(e)\%m], \cdots, B[h_k(e)\%m]$ to 1.
To process a membership query of whether element $e$ is in $S$, BF returns true if all corresponding $k$ bits are 1 (\ie, returns $\wedge_{i=1}^{k}B[h_i(e)\%m]$).
BF has no false negatives (FNs), \ie, it never says that $e$$\notin$$S$ when actually $e$ $\in$ $S$.
However, BF has false positives (FPs), \ie, it may say that $e\in S$ when actually $e\notin S$ with a certain probability.
Note that BF does not support element deletion.
CBF overcomes this shortcoming by replacing each bit in BF by a counter.
Given a set of elements, CBF first constructs an array $C$ of $m$ counters, where each counter is initialized to 0.
For each element $e$ in $S$, for each $1\leqslant i\leqslant k$, CBF increments $C[h_i(e)\%m]$ by 1.
To process a membership query of whether element $e$ is in set $S$, CBF returns true if all corresponding $k$ counters are at least 1 (\ie, returns $\wedge_{i=1}^{k}(C[h_i(e)\%m] \geqslant 1)$).
To delete an element $e$ from $S$, for each $1\leqslant i\leqslant k$, CBF decrements $C[h_i(e)\%m]$ by 1.

\vspace{-0.03in}
\subsection{Proposed Approach} \vspace{-0.03in}
\label{subsec:ProposedApproach}
In this paper, we propose a \emph{Shifting Bloom Filter} (ShBF) framework for representing and querying sets.
Let $h_1(.)$, $\cdots$, $h_k(.)$ be $k$ independent hash functions with uniformly distributed outputs.
In the construction phase, ShBF first constructs an array $B$ of $m$ bits, where each bit is initialized to 0.
We observe that in general a set data structure needs to store two types of information for each element $e$: (1) \emph{existence information}, \ie, whether $e$ is in a set, and (2) \emph{auxiliary information}, \ie, some additional information such as $e$'s counter (\ie, multiplicity) or which set that $e$ is in.
For each element $e$, we encode its existence information in $k$ hash values $h_1(e)\%m, \cdots, h_k(e)\%m$, and its auxiliary information in an offset $o(e)$.
Instead of, or in addition to, setting the $k$ bits at locations $h_1(e)\%m, \cdots, h_k(e)\%m$ to 1, we set the bits at locations $(h_1(e)+o(e))\%m, \cdots, (h_k(e)+o(e))\%m$ to 1.
\textcolor{bluer}{For different set queries, the offset has different values.
%
%
In the query phase, to query an element $e$, we first calculate the following $k$ locations: $h_1(e)\%m, \cdots, h_k(e)\%m$.
Let $c$ be the maximum value of all offsets.
For each $1 \leqslant i \leqslant k$, we first read the $c$ bits $B[h_i(e)\%m], B[(h_i(e)+1)\%m], \cdots, B[(h_i(e)+c-1)\%m]$ and then calculate the existence and auxiliary information about $e$ by analyzing where 1s appear in these $c$ bits.
To minimize the number of memory accesses, we extend the number of bits in ShBF to $m+c$; thus, we need $k \lceil \frac{c}{w} \rceil$ number of memory accesses in the worst case, where $w$ is the word size.}
%
%
Figure \ref{draw:model} illustrates our ShBF framework.

\begin{figure}[htbp]
\centering
\vspace{-0.05in}
\includegraphics[width=0.38\textwidth]{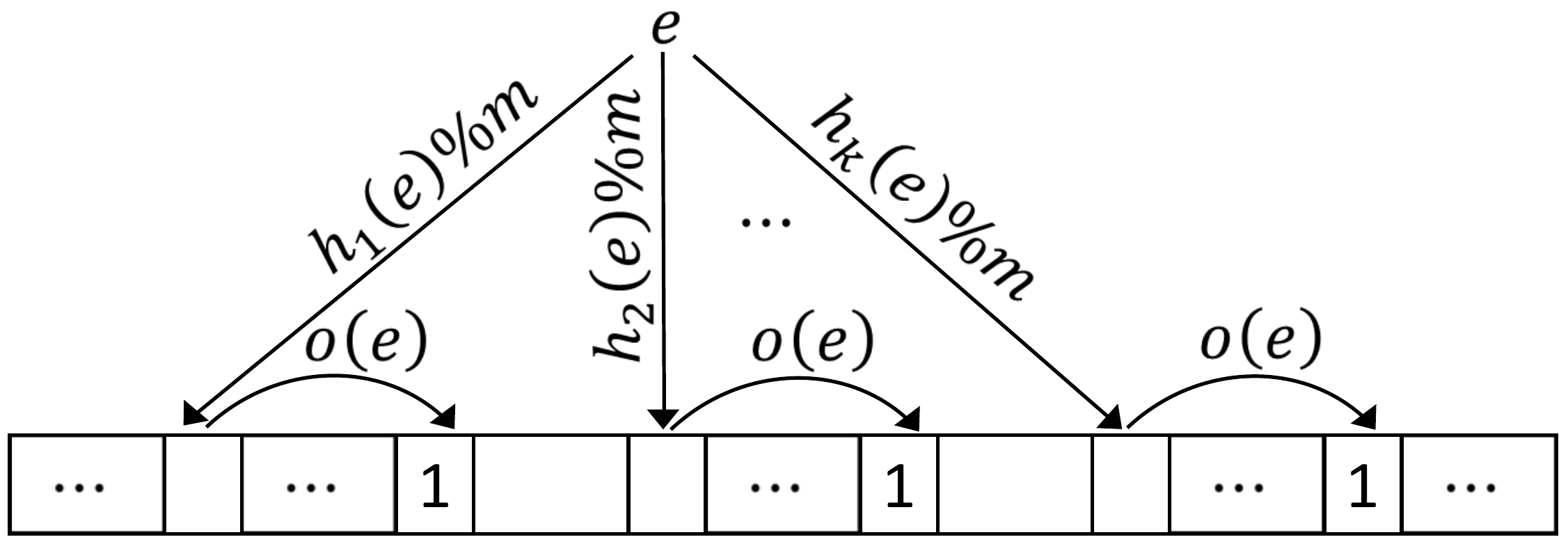}
\vspace{-0.03in}
\caption{Shifting Bloom Filter framework.} \label{draw:model}
\vspace{-0.03in}
\end{figure}

We demonstrate the effectiveness of ShBF using three types of popular set queries: membership, association, and multiplicity queries.

\vspace{-0.06in}
\subsubsection{Membership Queries}
Such queries only deal with the existence information of each element, which is encoded in $k$ random positions in array $B$.
To leverage our ShBF framework, we treat $k/2$ positions as the existence information and the other $k/2$ positions as the auxiliary information, assuming $k$ is an even number for simplicity.
Specifically, the offset function $o(.)=h_{\frac{k}{2}+1}(.)\%(\overline{w}-1)+1$, where $h_{\frac{k}{2}+1}(.)$ is another hash function with uniformly distributed outputs and  $\overline{w}$ is a function of machine word size $w$.
In the \textbf{construction phase}, for each element $e \in S$, we set both the $k/2$ bits $B[h_1(e)\%m], \cdots, B[h_{\frac{k}{2}}(e)\%m]$ and the $k/2$ bits $B[h_1(e)\%m+o(e)], \cdots, B[h_{\frac{k}{2}}(e)\%m+o(e)]$ to 1.
In the \textbf{query phase}, for an element $e$, if all these $k$ bits are 1, then we output $e \in S$; otherwise, we output $e \notin S$.
In terms of false positive rate (FPR), our analysis shows that ShBF is very close to BF with $k$ hash functions.
In terms of performance, ShBF is about two times faster than BF because of two main reasons.
First, ShBF reduces the computational cost by almost half because the number of hash functions that ShBF needs to compute is almost the half of what BF needs to compute.
Second, ShBF reduces the number of memory accesses by half because although both ShBF and BF write $k$ bits into the array $B$, when querying element $e$, by one memory access, ShBF obtains two bits about $e$ whereas BF obtains only one bit about $e$.

\vspace{-0.07in}
\subsubsection{Association Queries}
For this type of queries with two sets $S_1$ and $S_2$, for elements in $S_1\cup S_2$, there are three cases: (1) $e\in S_1-S_2$, (2) $e\in S_1\cap S_2$, and (3) $e\in S_2-S_1$.
For the first case, \ie, $e\in S_1-S_2$, the offset function $o(e)=0$.
For the second case, \ie, $e\in S_1\cap S_2$, the offset function $o(e)=o_1(e)=h_{k+1}(e) \%((\overline{w}-1)/2)+1$, where $h_{k+1}(.)$ is another hash function with uniformly distributed outputs and  $\overline{w}$ is a function of machine word size $w$.
For the third case, \ie, $e\in S_2-S_1$, the offset function $o(e)=o_2(e)=o_1(e)+h_{k+2}(e) \%((\overline{w}-1)/2)+1$, where $h_{k+2}(.)$ is yet another hash function with uniformly distributed outputs.
In the \textbf{construction phase}, for each element $e\in S_1\cup S_2$, we set the $k$ bits $B[h_1(e)\%m+o(e)], \cdots, B[h_k(e)\%m+o(e)]$ to 1 using an appropriate value of $o(e)$ as just described for the three cases.
In the \textbf{query phase}, given an element $e\in S_1\cup S_2$ , for each $1 \leqslant i \leqslant k$, we read the 3 bits $B[h_i(e)\%m]$, $B[h_i(e)\%m+o_1(e)]$, and $B[h_i(e)\%m+o_2(e)]$.
If all the $k$ bits $B[h_1(e)\%m], \cdots, B[h_k(e)\%m]$ are 1, then $e$ may belong to $S_1-S_2$.
If all the $k$ bits $B[h_1(e)\%m+o_1(e)], \cdots, B[h_k(e)\%m+o_1(e)]$ are 1, then $e$ may belong to $S_1 \cap S_2$.
If all the $k$ bits $B[h_1(e)\%m+o_2(e)], \cdots, B[h_k(e)\%m+o_2(e)]$ are 1, then $e$ may belong to $S_2-S_1$.
There are a few other possibilities that we will discuss later in Section \ref{sec:AssociationQueries}, that ShBF takes into account when answering the association queries.
In comparison, the standard BF based association query scheme, namely iBF, constructs a BF for each set.
In terms of accuracy, iBF is prone to false positives whenever it declares an element $e\in S_1\cup S_2$ in a query to be in $S_1\cap S_2$, whereas ShBF achieves an FPR of zero.
In terms of performance, ShBF is almost twice as fast as iBF because iBF needs $2k$ hash functions and $2k$ memory accesses per query, whereas ShBF needs only $k+2$ hash functions and $k$ memory accesses per query.
\vspace{-0.02in}
\subsubsection{Multiplicity Queries}
For multiplicity queries, for each element $e$ in a multi-set $S$, the offset function $o(.)=c(e)-1$ where $c(e)$ is $e$'s counter (\ie, the number of occurrences of $e$ in $S$).
In the \textbf{construction phase}, for each element $e$, we set the $k$ bits $B[h_1(e)\%m+c(e)-1], \cdots, B[h_k(e)\%m+c(e)-1]$ to 1.
In the \textbf{query phase}, for an element $e$, for each $1 \leqslant i \leqslant k$, we read the $c$ bits $B[h_i(e)\%m], B[h_i(e)\%m+1], \cdots, B[h_i(e)\%m+c-1]$, where $c$ is the maximum number of occurrences that an element can appear in $S$.
For these $ck$ bits, for each $1 \leqslant j \leqslant c$, if all the $k$ bits $B[h_1(e)\%m+j-1], \cdots, B[h_k(e)\%m+j-1]$ are 1, then we output $j$ as one possible value of $c(e)$.
Due to false positives, we may output multiple possible values.

\vspace{-0.02in}
\subsection{Novelty and Advantages over Prior Art}\vspace{-0.02in}
The key novelty of ShBF is on encoding the auxiliary information of a set element in its location by the use of offsets.
In contrast, prior BF based set data structures allocate additional memory to store such auxiliary information.

To evaluate our ShBF framework in comparison with prior art, we conducted experiments using real-world network traces.
Our results show that ShBF significantly advances the state-of-the-art on all three types of set queries: membership, association, and multiplicity.
For membership queries, in comparison with the standard BF, ShBF has about the same FPR but is about 2 times faster; in comparison with 1MemBF \cite{onemem}, which represents the state-of-the-art in membership query BFs, ShBF has 10\% $\sim$ 19\% lower FPR and $1.27\sim 1.44$ times faster query speed.
%
%
%
For association queries, in comparison with iBF, ShBF has 1.47 times higher probability of a clear answer, and has 1.4 times faster query speed.
%
For multiplicity queries, in comparison with Spectral BF \cite{spectralBF}, which represents the state-of-the-art in multiplicity query BFs, ShBF has 1.45 $\sim$ 1.62 times higher correctness rate and the query speeds are comparable.
%


%% file: relatedwork.tex
\section{Related Work} \label{sec:relatedwork}
We now review related work on the three types of set queries: membership, association, and multiplicity queries, which are mostly based on Bloom Filters.
Elaborate surveys of the work on Bloom Filters can be found in \cite{BFSurvey2012, BFSurvey9C, BFsurvey05, BFSurvey2004}.

\vspace{-0.03in}
\subsection{Membership Queries}
Prior work on membership queries focuses on optimizing BF in terms of the number of hash operations and the number of memory accesses.
Fan \etal proposed the Cuckoo filter and found that it is more efficient in terms of space and time compared to BF \cite{cuckoo}.
This improvement comes at the cost of non-negligible probability of failing when inserting an element. 
To reduce the number of hash computation, Kirsch \etal proposed to use two hash functions $h_1(.)$ and $h_2(.)$ to simulate $k$ hash functions $(h_1(.)+i*h_2(.)) \%m$, where $(1 \leqslant i\leqslant k)$; but the cost is increased FPR \cite{lessHashBF}.
To reduce the number of memory accesses, Qiao \etal proposed to confine the output of the $k$ hash functions within certain number of machine \texttt{word}s, which reduces the number of memory accesses during membership queries; but the cost again is increased FPR \cite{onemem}.
In contrast, ShBF reduces the number of hash operations and memory access by about half while keeping FPR about the same as BF.

\vspace{-0.03in}
\subsection{Association Queries}
Prior work on association queries focuses on identifying the set, among a group of pair-wise disjoint sets, to which an element belongs.
A straightforward solution is iBF, which builds one BF for each set.
To query an element, iBF generates a membership query for each set's BF and finds out which set(s) the unknown element is in.
This solution is used in the Summary-Cache Enhanced ICP protocol \cite{webcaching}.
Other notable solutions include kBF \cite{kbf}, Bloomtree \cite{bloomtree}, Bloomier \cite{bloomier}, Coded BF \cite{codedBF}, Combinatorial BF \cite{combineBF}, SVBF \cite{SVBF}.
A common shortcoming of all existing schemes is that if any pair of sets in the group of sets is not disjoint, these schemes do not function correctly.
%
%
In contrast, ShBF does not require the sets to be disjoint.

\vspace{-0.03in}
\subsection{Multiplicity Queries}
BF cannot process multiplicity queries because it only tells whether an element is in a set.
Spectral BF, which was proposed by Cohen and Matias, represents the state-of-the-art scheme for multiplicity queries \cite{spectralBF}.
There are three versions of Spectral BF.
The first version proposes some modifications to CBF to record the multiplicities of elements.
The second version increases only the counter with the minimum value when inserting an element.
This version reduces FPR at the cost of not supporting updates.
The third version minimizes space for counters with a secondary spectral BF and auxiliary tables, which makes querying and updating procedures time consuming and more complex.
Aguilar-Saborit \etal proposed Dynamic Count Filters (DCF), which combines the ideas of spectral BF and CBF, for multiplicity queries \cite{Dcountingfilter}.
DCF uses two filters: the first filter uses fixed size counters and the second filter dynamically adjusts counter sizes.
The use of two filters degrades query performance.
\textcolor{bluer}{Another well-known scheme for multiplicity queries is the Count-Min (CM) Sketch \cite{CMsketch}.
We will describe the details of CM sketch and how our ShBF framework can be applied to CM sketch in Section \ref{sec:CMsketch}.}

%% file: membership.tex
\vspace{0.05in}
\section{Membership Queries} \label{sec:MembershipQueries} \postsec
In this section, we first present the construction and query phases of ShBF for membership queries.
Membership queries are the ``traditional'' use of a BF.
We use ShBF$_\text{M}$ to denote the ShBF scheme for membership queries.
Second, we describe the updating method of ShBF$_\text{M}$.
Third, we derive the FPR formula of ShBF$_\text{M}$.
Fourth, we compare the performance of ShBF$_\text{M}$ with that of BF.
Last, we present a generalization of ShBF$_\text{M}$.
Table \ref{table:symbols} summarizes the symbols and abbreviations used in this paper.

\vspace{-0.2in}
\begin{table}[htbp]
\centering\caption{Symbols \& abbreviations used in the paper}
\begin{tabular}{|c|l|}
\hline
\textbf{Symbol}&\textbf{Description}\\
\hline
$m$& size of a Bloom Filter\\
\hline
$n$& \# of elements of a Bloom Filter\\
\hline
$k$& \# of hash functions of a Bloom Filter\\
\hline
$k_{opt}$& the optimal value of $k$\\
\hline
$S$& a set\\
\hline
$e$& one element of a set\\
\hline
$u$& one element of a set\\
\hline
$h_i(s)$& the $i$-th hash function\\
\hline
FP& false positive\\
\hline
FPR& false positive rate\\
\hline
$f$& the FP rate of a Bloom Filter\\
\hline
\multirow{2}{*}{$p'$}&the probability that one bit is still 0   \\& after inserting all elements into BF \\
\hline
BF& standard Bloom Filter\\
\hline
iBF& \tabincell{l}{individual BF: the solution that builds one \\individual BF per set}\\
\hline
ShBF& Shifting Bloom Filters\\
\hline
ShBF$_\text{M}$& Shifting Bloom Filters for membership qrs.\\
\hline
ShBF$_\text{A}$& Shifting Bloom Filters for association qrs.\\
\hline
ShBF$_\times$& Shifting Bloom Filters for multiplicities qrs.\\
\hline
Qps & queries per second\\
\hline
multi-set &  \tabincell{l}{a generalization of the notion of a set in\\which members can appear more than once}\\
\hline
$o(.)$& \tabincell{l}{offset(.), referring to the offset value for a \\ given input}\\
\hline
$w$& \# of bits in a machine word\\
\hline
$\overline{w}$& \tabincell{l}{the maximum value of offset(.) for \\ membership query of a single set}\\
\hline
$c$ & \tabincell{l}{the maximum number of times \\ an element can occur in a multi-set}\\
\hline
\end{tabular}
\label{table:symbols}
\vspace{-0.15in}
\end{table}

\subsection{ShBF$_\textbf{M}$ -- Construction Phase} \postsub\label{subsec:ShBFMConstructionPhase}
The construction phase of ShBF$_\text{M}$ proceeds in three steps.
Let $h_1(.), h_2(.), \cdots, h_{\frac{k}{2}+1}(.)$ be $\frac{k}{2}+1$ independent hash functions with uniformly distributed outputs.
First, we construct an array $B$ of $m$ bits, where each bit is initialized to 0.
Second, to store the existence information of an element $e$ of set $S$, we calculate $\frac{k}{2}$ hash values $h_1(e)\%m$,$h_2(e)\%m$,$\cdots$,$h_\frac{k}{2}(e)\%m$.
To leverage our ShBF framework, we also calculate the offset values for the element $e$ of set $S$ as the auxiliary information for each element, namely $o(e)=h_{\frac{k}{2}+1}(e)\hspace{0.02in}\%\hspace{0.02in}(\overline{w}-1)+1$.
We will later discuss how to choose an appropriate value for $\overline{w}$.
Third, we set the $\frac{k}{2}$ bits $B[h_1(e)\%m]$, $\cdots$, $B[h_\frac{k}{2}(e)\%m]$ to 1 and the other $\frac{k}{2}$ bits $B[h_1(e)\%m+o(e)], \cdots, B[h_\frac{k}{2}(e)\%m+o(e)]$ to 1.
Note that $o(e)\neq 0$ because if $o(e)= 0$, the two bits $B[h_i(e)\%m]$ and $B[h_i(e)\%m+o(e)]$ are the same bits for any value of $i$ in the range $1 \leqslant i \leqslant \frac{k}{2}$.
For the construction phase, the maximum number of hash operations is $\frac{k}{2}+1$.
Figure \ref{draw:ShBF_S} illustrates the construction phase of ShBF$_\text{M}$.

\begin{figure}[htbp]
\centering
\includegraphics[width=0.38\textwidth]{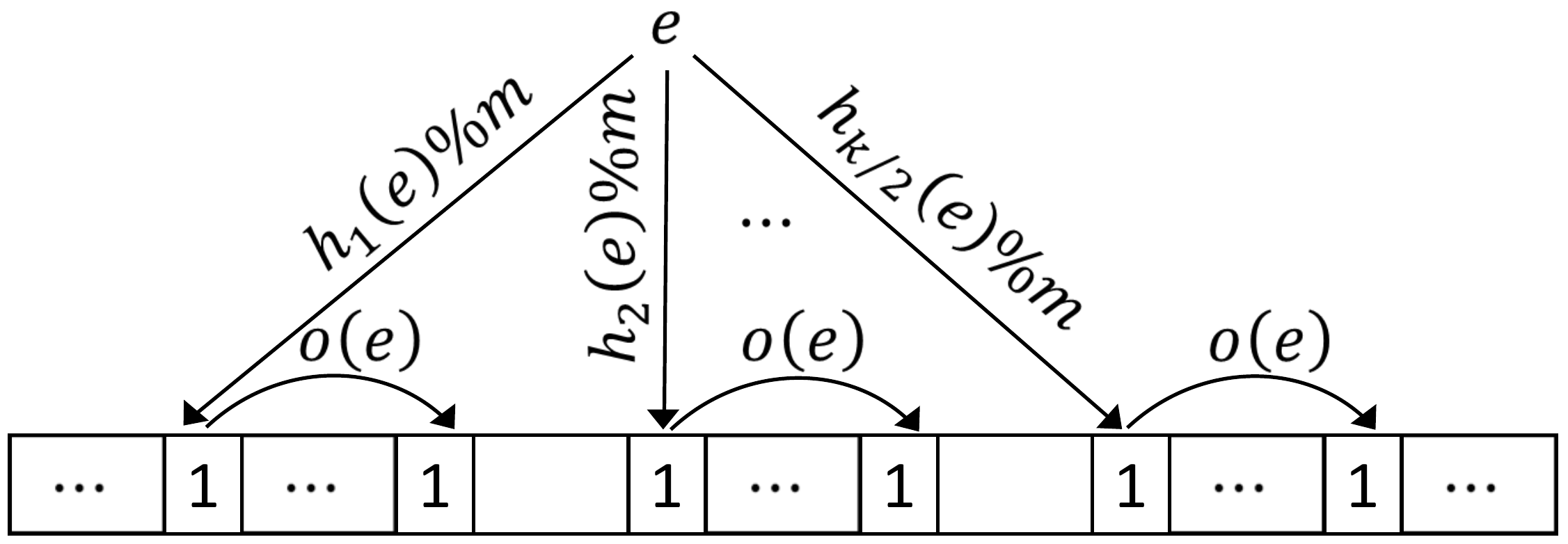}
\prefigcaption
\caption{Illustration of ShBF$_\text{M}$ construction phase.}
\label{draw:ShBF_S}
\postfig
\end{figure}

We now discuss how to choose a proper value for $\overline{w}$ so that for any $1 \leqslant i \leqslant \frac{k}{2}$, we can access both bits $B[h_i(e)\%m]$ and $B[h_i(e)\%m+o(e)]$ in one memory access.
Note that modern architecture like x86 platform CPU can access data starting at any byte, \ie, can access data aligned on any boundary, not just on word boundaries.
Let $B[h_i(e)\%m]$ be the $j$-th bits of a byte where $1 \leqslant j \leqslant 8$.
To access bit $B[h_i(e)\%m]$, we always need to read the $j-1$ bits before it.
To access both bits $B[h_i(e)\%m]$ and $B[h_i(e)\%m+o(e)]$ in one memory access, we need to access $j-1+\overline{w}$ bits in one memory access.
Thus, $j-1+\overline{w} \leqslant w$, which means $\overline{w} \leqslant w+1-j$.
When $j=8$, $w+1-j$ has the minimum value of $w-7$.
Thus, we choose $\overline{w} \leqslant w-7$ as it guarantees that we can read both bits $B[h_i(e)\%m]$ and $B[h_i(e)\%m+o(e)]$ in one memory access.

\presub
\subsection{ShBF$_\textbf{M}$ -- Query Phase} \postsub
Given a query $e$, we first read the two bits $B[h_1(e)\%m]$ and $B[h_1(e)\%m+o(e)]$ in one memory access.
If both bits are 1, then we continue to read the next two bits $B[h_2(e)\%m]$ and $B[h_2(e)\%m+o(e)]$ in one memory access; otherwise we output that $e \notin S$ and the query process terminates.
If for all $1\leqslant i \leqslant \frac{k}{2}$, $B[h_i(e)\%m]$ and $B[h_i(e)\%m+o(e)]$ are 1, then we output $e \in S$.
For the query phase, the maximum number of memory accesses is $\frac{k}{2}$.

\presub
\subsection{ShBF$_\textbf{M}$ -- Updating} \postsub
Just like BF handles updates by replacing each bit by a counter, we can extend ShBF$_\textbf{M}$ to handle updates by replacing each bit by a counter.
We use CShBF$_\text{M}$ to denote this counting version of ShBF$_\textbf{M}$.
Let $C$ denote the array of $m$ counters.
To insert an element $e$, instead of setting $k$ bits to 1, we increment each of the corresponding $k$ counters by 1; that is, we increment both $C[h_i(e)\%m]$ and $C[h_i(e)\%m+o(e)]$ by 1 for all $1 \leqslant i \leqslant \frac{k}{2}$.
To delete an element $e\in S$, we decrement both $C[h_i(e)\%m]$ and $C[h_i(e)\%m+o(e)]$ by 1 for all $1 \leqslant i \leqslant \frac{k}{2}$.
%
{\color{bluer} In most applications, 4 bits for a counter are enough. Therefore, we can further reduce the number of memory accesses for updating CShBF$_\text{M}$.} Similar to the analysis above, if we choose $\overline{w} \leqslant \lfloor \frac{w-7}{z} \rfloor$ where $z$ is the number of bits for each counter, we can guarantee to access both $C[h_i(e)\%m]$ and $C[h_i(e)\%m+o(e)]$ in one memory access.
{\color{bluer} Consequently, one update of CShBF$_\text{M}$ needs only $k/2$ memory accesses.
}

Due to the replacement of bits by counters, array $C$ in CShBF$_\text{M}$ uses much more memory than array $B$ in ShBF$_\textbf{M}$.
To have the benefits of both fast query processing and small memory consumption, we can maintain both ShBF$_\textbf{M}$ and CShBF$_\text{M}$, but store array $B$ in fast SRAM and array $C$ in DRAM.
Note that SRAM is at least an order of magnitude faster than DRAM.
Array $B$ in fast SRAM is for processing queries and array $C$ in slow DRAM is only for updating.
After each update, we synchronize array $C$ with array $B$.
The synchronization is quite straightforward: when we insert an element, we insert it to both array $C$ and $B$; when we delete an element, we first delete it from $C$, if there is at least one of the $k$ counters becomes 0, we clear the corresponding bit in $B$ to 0.

\presub
\subsection{ShBF$_\textbf{M}$ -- Analysis} \postsub
We now calculate the FPR of ShBF$_\text{M}$, denoted as $f_{\text{ShBF}_\text{M}}$.
Then, we calculate the minimum value of $\overline{w}$ so that ShBF$_\text{M}$ can achieve almost the same FPR as BF.
Last, we calculate the optimum value of $k$ that minimizes $f_{\text{ShBF}_\text{M}}$.

\presub
\subsubsection{False Positive Rate}
We calculate the false positive rate of ShBF$_\text{M}$ in the following theorem.
\begin{thm}\label{thm:pfpShBFM}
The FPR of ShBF$_\text{M}$ for a set of $n$ elements is calculated as follows:
\begin{equation}
\label{equ:pfpShBFM}
f_{\emph{ShBF}_\text{M}} \approx (1-p)^\frac{k}{2} \left( 1-p +\dfrac{1}{\overline{w}-1} p^2 \right)^\frac{k}{2}
\end{equation}
where $p=e^{\frac{-nk}{m}}$.
\end{thm}
\begin{proof}
Let $p'$ represent the probability that one bit (suppose it is at position $i$) in the filter $B$ is still 0 after inserting information of all $n$ elements.
For an arbitrary element $e$, if $h_i(e)\%m$ does not point to $i$ or $i-o(e)$, where $o(e)=h_\frac{k}{2}(e)\%(\overline{w}-1)+1$, then the bit at position $i$ will still be 0, thus $p'$ is given by the following equation.
\begin{equation}
\vspace{-0.05in}
p'=\left(\dfrac{m-2}{m}\right)^{\dfrac{kn}{2}}=\left(1-\dfrac{2}{m}\right)^{\dfrac{kn}{2}}
\label{valueofp'}
\end{equation}
When $m$ is large, we can use the identity $\sum\limits_{x}^{\infty}\left( 1-\dfrac{1}{x} \right)^{-x}=e$, to get the following equation for $p'$.
\begin{equation}
\vspace{-0.05in}
p'=\left(1-\dfrac{2}{m}\right)^{\dfrac{kn}{2}}=\left( \left(1-\dfrac{2}{m}\right)^{\dfrac{m}{2}}\right)^{\dfrac{kn}{m}}
\approx e^{-\dfrac{nk}{m}}
\label{valueofp}
\end{equation}

Let $X$ and $Y$ be the random variables for the event that the bit at position $h_i(.)$ and the bit at position $h_i(.) + h_{\frac{k}{2}+1}(.)$ is 1, respectively.
Thus, $P\left\{X\right\}=1-p'$.
Suppose we look at a hash pair $\langle h_i, h_{\frac{k}{2}}+1 \rangle$, we want to calculate $P\left\{XY\right\}$.
As $P\left\{XY\right\}=P\left\{X\right\}\times P\left\{Y|X\right\}$, next we calculate $P\left\{Y|X\right\}$.
There are $\overline{w}-1$ bits on the left side of position $h_i$.
The $1s$ in these $\overline{w}-1$ bits could be due to the first hash function in a pair and/or due to the second hash function in the pair .
In other words, event $X$ happens because a hash pair $\langle h_j$, $h_{\frac{k}{2}+1}\rangle $ sets the position $h_i$ to 1 during the construction phase.
When event $X$ happens, there are two cases:
\begin{enumerate}
\item The event $X1$ happens, \ie, the position $h_i$ is set to 1 by $h_{\frac{k}{2}+1}$, \ie, the left $\overline{w}-1$ bits cause $h_i$ to be 1, making $X$ and $Y$ independent. Thus, in this case $P\left\{Y\right\}=1-p'$.
\item The event $X2$ happens, \ie, the position $h_i$ is set to 1 by $h_j$. In this case,
As $P\left\{X1\right\}$ + $P\left\{X2\right\}$ =1, thus, $P\left\{Y|X\right\}=P\left\{Y|X, X1\right\}\times P\left\{X1\right\} + P\left\{Y|X, X2\right\}\times P\left\{X2\right\}$.
\end{enumerate}
Next, we compute $P\left\{X1\right\}$ and $P\left\{X2\right\}$.

As there are $\overline{w}-1$ bits on the left side of position $h_i$, there are $\overline{w}-1$ combinations, \ie, $\binom{\overline{w}-1}{1}=\overline{w}-1$.
Probability that any bit of the $\overline{w}-1$ bits is 1 is $1-p'$.
When one bit in the $\overline{w}-1$ bits is 1, probability that this bit sets the bit at location $h_i$ using the hash function $h_{\frac{k}{2}}+1$ to 1 is $\frac{1}{\overline{w}-1}$.
Therefore, $P\left\{X1\right\}=\binom{\overline{w}-1}{1} \times (1-p') \times \frac{1}{\overline{w}-1} = 1-p'$.
Consequently, $P\left\{X2\right\} = 1- P\left\{X1\right\}=p'$.
Again there are two cases:
\begin{enumerate}
\item If the bit which $h_{i}(x)$ points to is set to 1 by the left $1s$, X and Y are independent, and thus $P\left\{Y\right\}=\binom{\overline{w}-1}{1}\times (1-p')\times\frac{1}{\overline{w}-1}=1-p'$.

\item If the bit which $h_{i}(x)$ points to is not set to 1 by the left $1s$, then it must set one bit of the latter $\overline{w}-1$ bits to be 1. This case will cause one bit of the latter $\overline{w}-1$ bits after position $h_i$ to be 1. In this case, there are following two situations for the second hashing $h_i+h_{\frac{k}{2}+1}$:
    \begin{enumerate}
    \item when the second hash points to this bit, the probability is $\frac{1}{\overline{w}-1}\times1$;
    \item otherwise, the probability is   $(1-\frac{1}{\overline{w}-1})\times(1-p')$.
    \end{enumerate}
\end{enumerate}
When the second case above happens, $P\left\{Y|X, X2\right\}$ is given by the following equation.
\begin{equation}
P\left\{Y|X, X2\right\}=\dfrac{(1-p')(\overline{w}-2)}{\overline{w}-1}+\dfrac{1}{\overline{w}-1}=\left(1-\dfrac{\overline{w}-2}{\overline{w}-1}p'\right)
\end{equation}
Integrating the two cases, we can compute $P\left\{Y|X\right\}$ as follows.
\begin{equation}
\label{hi:k:2}
\begin{aligned}
P\left\{Y|X\right\} = (1-p')(1-p')+ \left(1-(1-p') \right)        \left(1-\dfrac{\overline{w}-2}{\overline{w}-1}p'\right)
\end{aligned}
\end{equation}

The probability that all the first hashes point to bits that are 1 is $(1-p')^{\frac{k}{2}}$.
The probability that the second hash points to a bit that is 1 is the $\frac{k}{2}$-th power of Equation \eqref{hi:k:2}.
Thus, the overall FPR of ShBF$_\text{M}$ is given by the following equation.
\vspace{-0.1in}
\begin{equation}
\begin{aligned}
f_{\text{ShBF}_\text{M}}=
&(1-p')^\frac{k}{2}  \left(  (1-p')(1-p')+p'\left(1-\dfrac{\overline{w}-2}{\overline{w}-1}p'\right)  \right)^\frac{k}{2}\\
=&(1-p')^\frac{k}{2}  \left( 1-p' +\dfrac{1}{\overline{w}-1}  p'^2 \right)^\frac{k}{2}  \\
\end{aligned}
\end{equation}

Note that when $\overline{w}\rightarrow\infty$, this formula becomes the formula of the FPR of BF.

Let we represent $e^{-\dfrac{nk}{m}}$ by $p$.
Thus, according to equation \ref{valueofp}, $p'\approx p$.
Consequently, we get:
\begin{equation*}
\label{fpofShBF_S}
\begin{aligned}
f_{\text{ShBF}_\text{M}}
\approx (1-p)^\frac{k}{2}  \left(  1-p +\dfrac{1}{\overline{w}-1}  p^2 \right)^\frac{k}{2}\qed
\end{aligned}
\end{equation*}
\end{proof}

Note that the above calculation of FPRs is based on the original Bloom's FPR formula \cite{BF1970}.
In 2008, Bose \etal pointed out that Bloom's formula \cite{BF1970} is slightly flawed and gave a new FPR formula \cite{bose2008false}.
Specifically, Bose \etal explained that the second independence assumption needed to derive $f_{Bloom}$ is too strong and does not hold in general, resulting in an underestimation of the FPR.
In 2010, Christensen \etal further pointed out that Bose's formula is also slightly flawed and gave another FPR formula \cite{ken2010false}.
Although Christensen's formula is \textcolor{bluer}{final}, it cannot be used to compute the optimal value of $k$, which makes the FPR formula practically not much useful.
Although Bloom's formula underestimates the FPR, both studies pointed out that the error of Bloom's formula is negligible.
Therefore, our calculation of FPRs is still based on Bloom's formula.

\vspace{-0.05in}
\presub
\subsubsection{Optimizing System Parameters}
\textbf{Minimum Value of $\overline{w}$:}
Recall that we proposed to use $\overline{w} \leqslant w-7$.
According to this inequation, $\overline{w} \leqslant 25$ for 32-bit architectures and $\overline{w} \leqslant 57$ for 64-bit architectures.
Next, we investigate the minimum value of $\overline{w}$ for ShBF$_\text{M}$ to achieve the same FPR with BF.
We plot $f_{\text{ShBF}_\text{M}}$ of ShBF$_\text{M}$ as a function of $\overline{w}$ in Figures \ref{fig:fp_w_k} and \ref{fig:fp_w_m}.
Figure \ref{fig:fp_w_k} plots $f_{\text{ShBF}_\text{M}}$ vs. $\overline{w}$ for $n=10000$, $m=100000$, and $k=$4, 8, and $12$ and Figure \ref{fig:fp_w_m} plots $f_{\text{ShBF}_\text{M}}$ vs. $\overline{w}$ for $n=10000$, $k=10$, and $m=100000, 110000,$ and $120000$.
The horizontal solid lines in these two figures plot the FPR of BF.
From these two figures, we observe that when $\overline{w}>20$, the FPR of ShBF$_\text{M}$ becomes almost equal to the FPR of BF.
Therefore, to achieve similar FPR as of BF, $\overline{w}$ needs to be larger than 20.
Thus, by using $\overline{w}=25$ for $32-$bit and $\overline{w}=57$ for $64-$bit architecture, ShBF$_\text{M}$ will achieve almost the same FPR as BF.
%
\begin{figure}[htbp]
\centering
\hspace{-0.11in}
\subfigure[$m=100000$, $n=10000$]{
{\includegraphics[width=0.48\columnwidth]{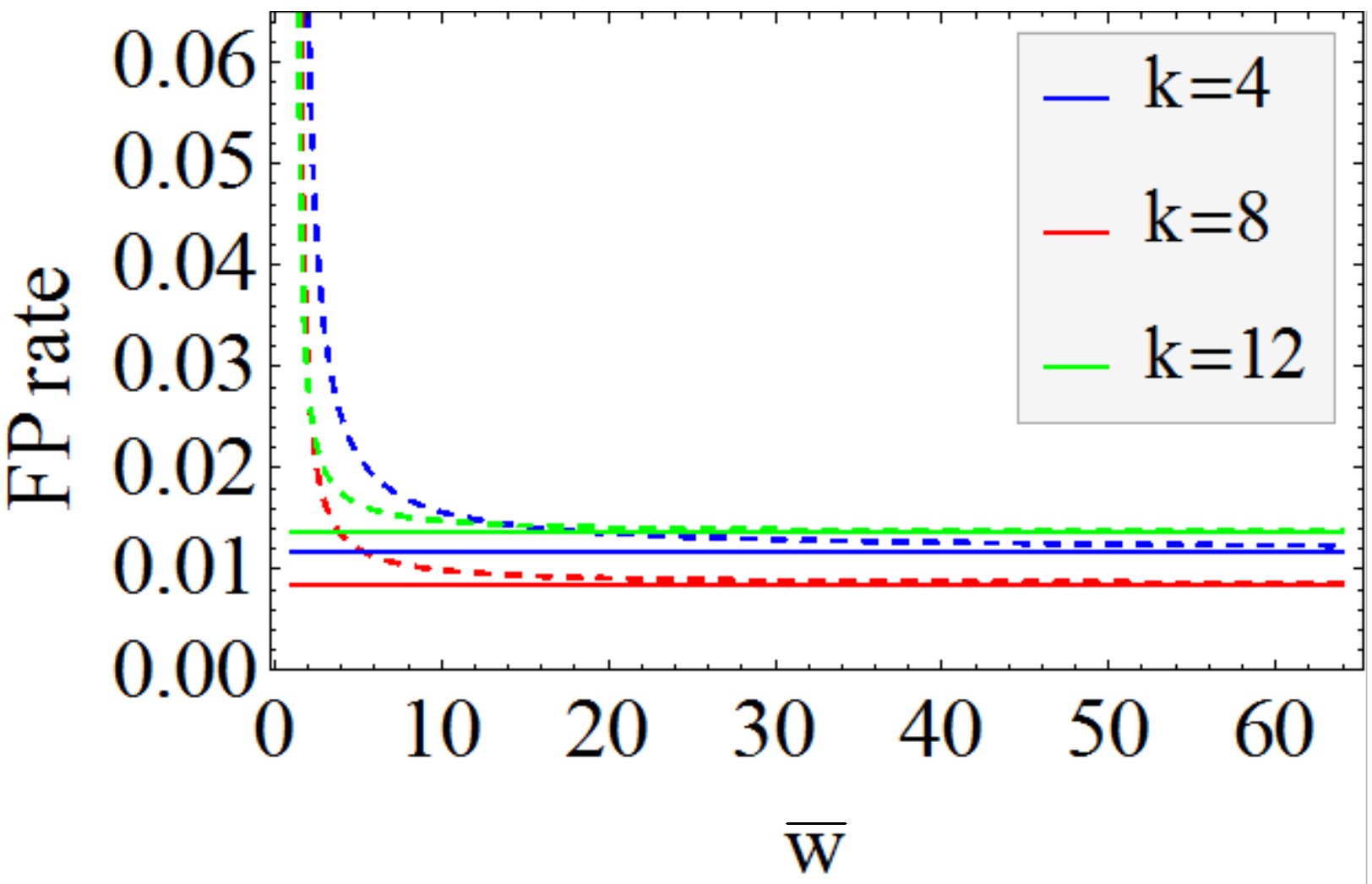}}
\label{fig:fp_w_k}}
\hspace{-0.11in}
\subfigure[$k=10$, $n=10000$]{
{\includegraphics[width=0.48\columnwidth]{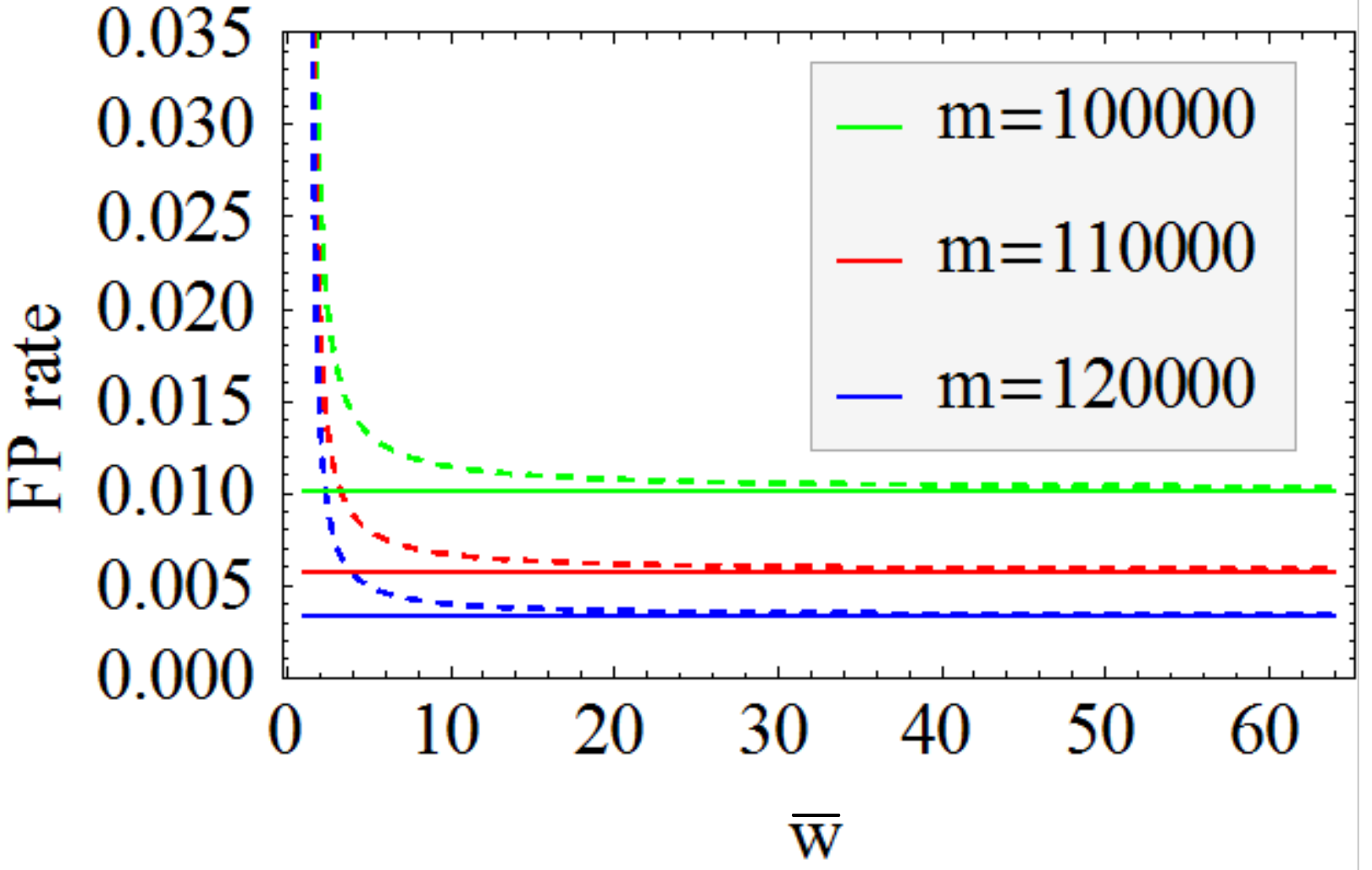}}
\label{fig:fp_w_m}}
\prefigcaption
\caption{FPR vs. $\overline{w}$.}
\label{fig:}
\postfig
\end{figure}

\textbf{Optimum Value of $k$:}
Now we calculate the value of $k$ that minimizes the FPR calculated in Equation \eqref{equ:pfpShBFM}.
The standard method to obtain the optimal value of $k$ is to differentiate Equation \eqref{equ:pfpShBFM} with respect to $k$, equate it to 0, \ie, $\frac{\partial}{\partial k}f_{\text{ShBF}_\text{M}}=0$, and solve this equation for $k$.
Unfortunately, this method does not yield a closed form solution for $k$.
Thus, we use standard numerical methods to solve the equation $\frac{\partial}{\partial k}f_{\text{ShBF}_\text{M}}=0$ to get the optimal value of $k$ for given values of $m,n,$ and $\overline{w}$.
For $\overline{w}=57$, differentiating Equation \eqref{equ:pfpShBFM} with respect to $k$ and solving for $k$ results in the following optimum value of $k$.
\begin{equation*}
\vspace{0.05in}
k_{\text{opt}}=0.7009\frac{m}{n}
\vspace{0.05in}
\end{equation*}

Substituting the value of $k_{\text{opt}}$ from the equation above into Equation \eqref{equ:pfpShBFM}, the minimum value of $f_{\text{ShBF}_\text{M}}$ is given by the following equation.
\begin{equation}
\label{equ:fminShBFM}
\vspace{0.05in}
f^{\text{min}}_{\text{ShBF}_\text{M}}=0.6204 ^ {\frac{m}{n}}
\vspace{0.05in}
\end{equation}

\presub
\subsection{Comparison of ShBF$_\text{M}$ FPR with BF FPR} \postsub
Our theoretical comparison of ShBF$_\text{M}$ and BF shows that the FPR of ShBF$_\text{M}$ is almost the same as that of BF.
Figure \ref{eva_fp_k} plots FPRs of ShBF$_\text{M}$ and  BF using Equations \eqref{equ:pfpShBFM} and \eqref{for:BF_FP}, respectively for $m=100000$ and $n=4000$, 6000, 8000, 10000, 12000.
The dashed lines in the figure correspond to ShBF$_\text{M}$ whereas the solid lines correspond to BF.
We observe from this figure that the sacrificed FPR of ShBF$_\text{M}$ in comparison with the FPR of  BF is negligible,  while the number of memory accesses and hash computations of ShBF$_\text{M}$ are half in comparison with  BF.
\vspace{-0.1in}
\begin{figure}[htbp]
\centering
\includegraphics[width=0.7\columnwidth]{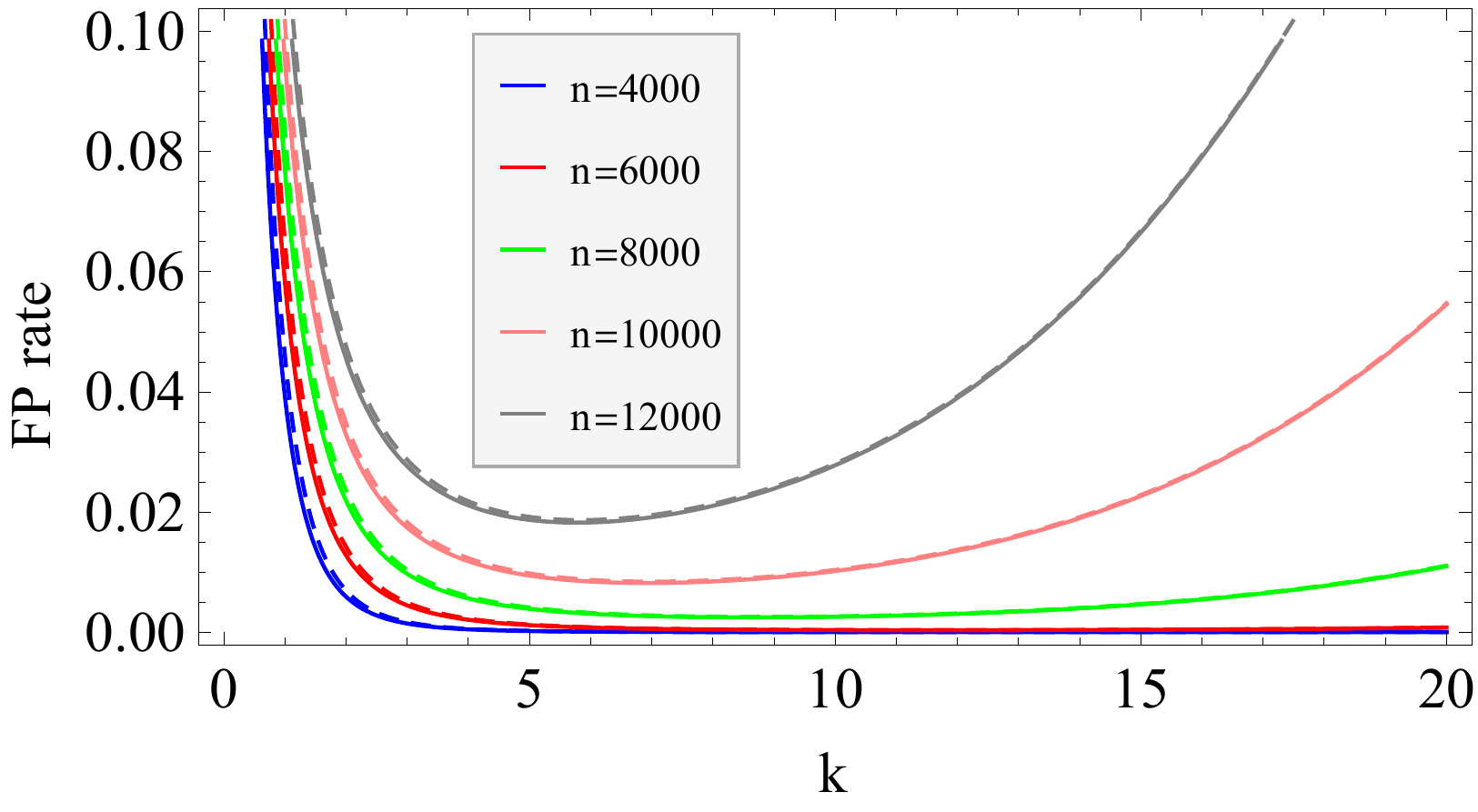}
\prefigcaption
\vspace{-0.1in}
\caption{ShBF$_\text{M}$ FPR vs. BF FPR.}
\vspace{-0.1in}
\postfig
\label{eva_fp_k}
\end{figure}

Next, we formally arrive at this result.
We calculate the minimum FPR of BF as we calculated for ShBF$_\text{M}$ in Equation \eqref{equ:fminShBFM} and show that the two FPRs are practically equal.
For a membership query of an element $u$ that does not belong to set $S$, just like ShBF$_\text{M}$, BF can also report true with a small probability, which is the FPR of BF and has been well studied in literature \cite{BF1970}.
It is given by the following equation.
\begin{equation}
f_{\text{BF}}=\left(1-\left(1-\frac{1}{m}\right)^{nk}\right)^k\approx\left(1-e^{-\frac{nk}{m}}\right)^k
\label{for:BF_FP}
\end{equation}
For given values of $m$ and $n$, the value of $k$ that minimizes $f_{\text{BF}}$ is $=\frac{m}{n}\ln2=0.6931\frac{m}{n}$.
Substituting this value of $k$ into Equation \eqref{for:BF_FP}, the minimum value of $f_{\text{BF}}$ is given by the following equation.
\begin{equation}
\label{equ:fminBF}
f^{\text{min}}_{\text{BF}}=\left(\frac{1}{2}\right)^{\left(\frac{m}{n}\ln2\right)}\approx 0.6185^{\frac{m}{n}}
\end{equation}

By comparing Equations \eqref{equ:fminShBFM} and \eqref{equ:fminBF}, we observe that the FPRs of ShBF$_\text{M}$ and BF are almost the same.
Thus, ShBF$_\text{M}$ achieves almost the same FPR as BF while reducing the number of hash computations and memory accesses by half.

\presub
{\color{greener}
\subsection{Generalization of ShBF$_\text{M}$} \postsub
As mentioned earlier, ShBF$_\text{M}$ reduces $k$ independent hash functions to $k/2+1$ independent hash functions.
Consequently, it calculates $k/2$ locations independently and remaining $k/2$ locations are correlated through the equation $h_i(e) + o_1(e)$ ($1\leqslant i\leqslant k/2$).
Carrying this construction strategy one step further, one could replace the first $k/2$ hash functions with $k/4$ independent hash functions and an offset $o_2(e)$, \ie, $h_j(e) + o_2(e)$ ($1\leqslant j\leqslant k/4$).
Continuing in this manner, one could eventually arrive at $log(k)+1$ hash functions.
Unfortunately, it is not trivial to calculate the FPR for this case because $log(k)$ is seldom an integer.
In this subsection, we simplify this $log$ method into a $linear$ method by first using a group of $\frac{k}{t+1}$ ($1\leqslant t\leqslant k-1$) hash functions to calculate $\frac{k}{t+1}$ hash locations and then applying shifting operation $t$ times on these hash locations.


Consider a group of hash function comprising of $t+1$ elements, \emph{i.e.}, $\langle h_1(x), h_2(x), \ldots , h_{t+1}(x) \rangle$.
After completing the construction phase using this group of hash functions, the probability that any given bit is 0 is $\dfrac{m-w}{m}+\dfrac{w-1}{m}\dfrac{w-2}{w-1} \ldots \dfrac{w-t-1}{w-t}=1-\dfrac{t+1}{m}$.
To insert $n$ elements, we need $\dfrac{nk}{t+1}$ such group insertion operations.
After completing the insertion, the probability $p'$ that one bit is still 0 is given by the following equation.
\begin{equation}
\label{g:p'}
\begin{aligned}
p'= \left(1-\dfrac{t+1}{m}\right)^{\dfrac{kn}{t+1}}\approx e^{-\dfrac{kn}{m}}
\end{aligned}
\end{equation}
Note that this probability formula is essentially $k$ times product of $e^{-\frac{n}{m}}$.
Thus, we can treat our ShBF$_\text{M}$ as a \emph{partitioned Bloom filter}, where the output of each hash function covers a distinct set of consecutive $\dfrac{w-1}{t}$ bits.
Setting $w=m$ makes this scheme partitioned Bloom filter.
The equations below calculate the FPR $f$ for this scheme.

\begin{equation}
\label{g1:ghbf}
\begin{aligned}
f = \left(1-p'\right)^{\dfrac{k}{t+1}} \times \left(f_{group}\right)^{\dfrac{k}{t+1}}
\end{aligned}
\end{equation}
where
\begin{equation}
\label{g1:group}
\begin{aligned}
f_{group} = &\dfrac{1}{t} \times \left(1-p'\right)^{2} \times \dfrac{\left(1-p'\right)^{t} - \left(1-\dfrac{w-1-t}{w-1} \times p'\right)^{t}}{\left(1-p'\right) - \left(1-\dfrac{w-1-t}{w-1} \times p'\right)} \\
&+ p' \times \left(1-\dfrac{w-1-t}{w-1} \times p'\right)^{t}
\end{aligned}
\end{equation}
Due to space limitations, we have moved the derivation of this equation to our extended version at arxiv.org.

When $t=1$, its false positive rate can be simplified as $f = \left(1-p'\right)^{\dfrac{k}{2}}$ *$\left(1-p'+\dfrac{1}{w-1} \times p'^{2}\right)^{\dfrac{k}{2}}$.
Similarly, when $w$ goes to infinity, FPR simplifies to $f=(1-p')^k$, which is the formula for FPR of a \emph{standard Bloom filter}.
}

%% file: generalization-now.tex
{\color{greener}
\subsection{False Positive Derivation}
\textbf{In our submission, we did not present this subsection due to space limitation. This is the false positive derivation details for interested reviewers.}

To query a non-existent object, the false positive occurs when the hash function group returns all $1s$, assuming the hash function group is $\langle h_i(x), o_1(x),\ldots , o_t(x)\rangle$ ($1\leqslant i \leqslant \frac{k}{t+1}$). 

1)  If the corresponding bit of $h_i(x)$ is not set by the left $w-1$ bits, then $h_i(x)$ must cause $t$ bits, in the following $w-1$ bits, to be set to 1.  There are $t+1$ situations in total. Note that the $r-$th $(r \in [0, t])$ situation represents that the corresponding bits of $r$ hash functions out of $t$ are set by $h_i(x)$, and another $t-r$ bits are not set by $h_i(x)$. Considering the corresponding bits of the hash function subset $<o_1(x), \ldots , o_t(x)>$, the probability that each of these bits is set by $h_i(x)$ is

\begin{equation}
\label{g:lamda1}
\begin{aligned}
\lambda _{1} = \dfrac{t}{w-1}
\end{aligned}
\end{equation}

and the probability that each of these bits is not set by $h_i(x)$ is

\begin{equation}
\label{g:lamda2}
\begin{aligned}
\lambda _{2} = \left(1-\dfrac{t}{w-1}\right) \times \left(1-p'\right)
\end{aligned}
\end{equation}

Therefore, the probability that the $r-$th situation occurs is $C_{t}^{r} \times \lambda _{1}^{r} \times \lambda _{2}^{t-r}$. The total probability that the $t+1$ situations occur is

\begin{equation}
\label{g:totProb1}
\begin{aligned}
f_{I} = \Sigma _{r=0}^{t} C_{t}^{r} \times \lambda _{1}^{r} \times \lambda _{2}^{t-r} = \left(\lambda _{1}+ \lambda _{2}\right)^{t}
\end{aligned}
\end{equation}

2) If the corresponding bit of $h_i(x)$ is set by the left $w-1$ bits, then the problem can be divided into $t$ situations, where each situation has a probability of $\dfrac{1}{t}$. Note that the $l'-$th $(l' \in [0, t-1])$ situation represents that the maximum number of bits in the current subgroup $\langle o_1(x), \ldots , o_t(x)\rangle$, which are set to 1 by the previous hash function group causing the corresponding bit of $h_i(x)$ to be set to 1. Since our hash function group adopts partitioned Bloom filter, each hash function in the previous hash function group can cause at most 1 bit to be set 1 in current hash function group. When the corresponding bit of $h_i(x)$ in the previous hash function group locates at the first $\dfrac{w-1}{t}$ bits, then the previous hash function group will cause 0 bit to be set 1 in current subgroup $\langle o_1(x), \ldots , o_t(x)\rangle$, because there is only one possibility that the last hash function in the previous group causes the bit of current $h_i(x)$ to be set to $1$. This is the $0-$th situation. Similarly, for the $l'-$th situation, if the corresponding bit of $h_i(x)$ in the previous hash function group locates at $\left[w-1+l' \times \dfrac{w-1}{t}, w-1+ (l'+1) \times \dfrac{w-1}{t}\right)$, there are at most $l'$ bits in the current hash function group set by the previous group, and at least $t-l'$ bits in the current hash function group not set by the previous group. Therefore, for the $l'-$th situation,  the probability that all the bits in the current hash function group are 1 is

\begin{equation}
\label{g:lth}
\begin{aligned}
f_{l'} &= \dfrac{1}{t} \times \left(\Sigma _{r'=0}^{l'} C_{l'}^{r'} \times \lambda _{1}^{r'} \times \lambda _{2}^{l'-r'}\right) \times \left(1-p'\right)^{t-l'} \\
&= \dfrac{1}{t} \left(\lambda _{1}+ \lambda _{2}\right)^{l'} \times \left(1-p'\right)^{t-l'} \\
\end{aligned}
\end{equation}

The total probability that all the $t$ situations happen is
\begin{equation}
\label{g:totProb2}
\begin{aligned}
f_{II} =& \Sigma_{l'=0}^{t-1} f_{l'} = \Sigma _{l'=0} ^{t-1}\dfrac{1}{t} \times \left(\left(\lambda _{1}+ \lambda _{2}\right)^{l'} \times \left(1-p'\right)^{t-l'}\right) \\
= &\dfrac{1}{t} \times (\left(\lambda _{1}+ \lambda _{2}\right)^{0} \times \left(1-p'\right)^{t-0}  + \ldots \\
 & + \dfrac{1}{t} \times \left(\lambda _{1}+ \lambda _{2}\right)^{t-1} \times \left(1-p'\right)^{t-\left(t-1\right)})
\end{aligned}
\end{equation}

Assuming $x = \Sigma _{l'=0} ^{t}\left(\left(\lambda _{1}+ \lambda _{2}\right)^{l'} \times \left(1-p'\right)^{t-l'}\right)$, we get the following equation

\begin{equation}
\label{g:sol}
\begin{aligned}
x \times \dfrac{\lambda _{1}+ \lambda _{2}} {1-p'} = x - \left(1-p'\right)^{t} + \left(\lambda _{1}+ \lambda _{2}\right)^{t}
\end{aligned}
\end{equation}

By solving equation \ref{g:sol}, we obtain the solution

\begin{equation}
\label{g:res}
\begin{aligned}
x = \dfrac{\left(1-p'\right)^{t} - \left(\lambda _{1}+ \lambda _{2}\right)^{t}} {\left(1-p'\right) - \left(\lambda _{1}+ \lambda _{2}\right)} \times \left(1-p'\right)
\end{aligned}
\end{equation}

Probability that 1) happens is $C _{w-1} ^{1} \times \left(\dfrac{1}{w-1}\right) \times \left(1-p'\right) = \left(1-p'\right)$, so the probability that 2) happens is $\left(1-\left(1-p'\right)\right)=p'$.

Combining 1) and 2), we know that when $h_i(x)$ is $1$, the probability that all the bits of the current group $\langle o_1(x),\ldots, o_t(x)\rangle$ are 1 is
\begin{equation}
\label{g:group}
\begin{aligned}
f_{group} = &\dfrac{1}{t} \times \left(1-p'\right)^{2} \times \dfrac{\left(1-p'\right)^{t} - \left(1-\dfrac{w-1-t}{w-1} \times p'\right)^{t}}{\left(1-p'\right) - \left(1-\dfrac{w-1-t}{w-1} \times p'\right)} \\
&+ p' \times \left(1-\dfrac{w-1-t}{w-1} \times p'\right)^{t}
\end{aligned}
\end{equation}
Probability that the corresponding bits of the first $\dfrac{k}{t+1}$ hash functions are $1$ is $\left(1-p'\right)^{\dfrac{k}{t+1}}$. Therefore, the false positive of the generalized ShBF$_\text{M}$ is
\begin{equation}
\label{g:ghbf}
\begin{aligned}
f = \left(1-p'\right)^{\dfrac{k}{t+1}} \times \left(f_{group}\right)^{\dfrac{k}{t+1}}
\end{aligned}
\end{equation}

Especially, when $t=1$, its false positive can be simplified as $f = \left(1-p'\right)^{\dfrac{k}{2}}$* \\ $\left(1-p'+\dfrac{1}{w-1} \times p'^{2}\right)^{\dfrac{k}{2}}$. Similarly, when $w$ goes to infinity, its false positive is $f=(1-p')^k$, which is the formula for \emph{standard Bloom filter}.

%% file: twosets.tex
\presec
\section{Association Queries} \postsec
\label{sec:AssociationQueries}
In this section, we first describe the construction and query phases of ShBF for association queries, \textcolor{reder}{which are also called membership test}.
We use ShBF$_\text{A}$ to denote the ShBF scheme for association queries.
Second, we describe the updating methods of ShBF$_\text{A}$.
Third, we derive the FPR of ShBF$_\text{A}$.
Last, we analytically compare the performance of ShBF$_\text{A}$ with that of iBF.

\presub
\subsection{ShBF$_\textbf{A}$ -- Construction Phase} \postsub
The construction phase of ShBF$_\text{A}$ proceeds in three steps.
Let $h_1(.), \cdots, h_k(.)$ be $k$ independent hash functions with uniformly distributed outputs.
Let $S_1$ and $S_2$ be the two given sets.
First, ShBF$_\text{A}$ constructs a hash table $T_1$ for set $S_1$ and a hash table $T_2$ for set $S_2$.
Second, it constructs an array $B$ of $m$ bits, where each bit is initialized to 0.
Third, for each element $e\in S_1$, to store its existence information, ShBF$_\text{A}$ calculates $k$ hash functions $h_1(e)\%m, \cdots, h_k(e)\%m$ and searches $e$ in $T_2$.
If it does not find $e$ in $T_2$, to store its auxiliary information, it sets the offset $o(e)=0$.
However, if it does find $e$ in $T_2$, to store its auxiliary information, it calculates the offset $o(e)$ as $o(e)=o_1(e)=h_{k+1}(e) \%((\overline{w}-1)/2)+1$, where $h_{k+1}(.)$ is a hash function with uniformly distributed output and $\overline{w}$ is a function of machine word size $w$, which we will discuss shortly.
Fourth, it sets the $k$ bits $B[h_1(e)\%m+o(e)], \cdots, B[h_k(e)\%m+o(e)]$ to 1.
Fifth, for each element $e\in S_2$, to store its existence information, ShBF$_\text{A}$ calculates the $k$ hash functions and searches it in $T_1$.
If it finds $e$ in $T_1$, it does not do anything because its existence as its auxiliary information have already been stored in the array $B$.
However, if it does not find $e$ in $T_1$, to store its auxiliary information, it calculates the offset $o(e)$ as $o(e)=o_2(e)=o_1(e)+h_{k+2}(e) \%((\overline{w}-1)/2)+1$, where $h_{k+2}(.)$ is also a hash function with uniformly distributed output.
Last, it sets the $k$ bits $B[h_1(e)\%m+o(e)], \cdots, B[h_k(e)\%m+o(e)]$ to 1.
To ensure that ShBF$_\text{A}$ can read $B[h_i(e) \%m]$, $B[h_i(e) \%m +o_1(e)]$, and $B[h_i(e) \%m +o_2(e)]$ in a single memory access when querying, we let $\overline{w}\leqslant w-7$.
We derived this condition $\overline{w}\leqslant w-7$ earlier at the end of Section \ref{subsec:ShBFMConstructionPhase}.
As the maximum value of $h_i(e)\%m + o_2(e)$ can be equal to $m+\overline{w}-2$, we append the $m$-bit array $B$ with $\overline{w}-2$ bits.

\presub
\subsection{ShBF$_\textbf{A}$ -- Query Phase} \postsub\label{subsec:ShBFAQueryPhase}

\textit{We assume that the incoming elements always belong to  $S_1\cup S_2$ in the load balance application\footnote{The application is mentioned in the first paragraph of Introduction Section.} for convenience.}
To query an element $e \in S_1\cup S_2$, ShBF$_\text{A}$ finds out which sets the element $e$ belongs to in the following three steps.
First, it computes $o_1(e)$, $o_2(e)$, and the $k$ hash functions $h_i(e)\%m$ ($1\leqslant i\leqslant k$).
Second, for each $1 \leqslant i \leqslant k$, it reads the $3$ bits $B[h_i(e)\%m]$, $B[h_i(e)\%m+o_1(e)]$, and $B[h_i(e)\%m+o_2(e)]$.
Third, for these $3k$ bits, if all the $k$ bits $B[h_1(e)\%m], \cdots, B[h_k(e)\%m]$ are 1, $e$ may belong to $S_1-S_2$.
In this case, ShBF$_\text{A}$ records (but does not yet declare) $e \: \overline{\in} \: S_1-S_2$.
Similarly, if all the $k$ bits $B[h_1(e)\%m+o_1(e)], \cdots, B[h_k(e)\%m+o_1(e)]$ are 1, $e$ may belong to $S_1\cap S_2$ and ShBF$_\text{A}$ records $e \: \overline{\in} \: S_1\cap S_2$.
Finally, if all the $k$ bits $B[h_1(e)\%m+o_2(e)], \cdots, B[h_k(e)\%m+o_2(e)]$ are 1, $e$ may belong to $S_2 - S_1$ and ShBF$_\text{A}$ records $e\, \overline{\in} \: S_2-S_1$.

Based on what ShBF$_\text{A}$ recorded after analyzing the $3k$ bits, there are following 7 outcomes.
If ShBF$_\text{A}$ records that:
\begin{enumerate}
\vspace{-0.05in}
\item only $e\: \overline{\in} \: S_1-S_2$, it declares that $e$ belongs to $S_1-S_2$.
\vspace{-0.05in}
\item only $e\: \overline{\in} \: S_1\cap S_2$, it declares that $e$ belongs to $S_1\cap S_2$.
\vspace{-0.05in}
\item only $e\: \overline{\in} \: S_2- S_1$, it declares that $e$ belongs to $S_2-S_1$.
\vspace{-0.05in}
\item both $e\: \overline{\in} \: S_1-S_2$ and $e\: \overline{\in} \: S_1\cap S_2$, it declares that $e$ belongs to $S_1$ but is unsure whether or not it belongs to $S_2$.
\vspace{-0.05in}
\item both $e\: \overline{\in} \: S_2-S_1$ and $e\: \overline{\in} \: S_1\cap S_2$, it declares that $e$ belongs to $S_2$ but is unsure whether or not it belongs to $S_1$.
\vspace{-0.05in}
\item both $e\: \overline{\in} \: S_1-S_2$ and $e\: \overline{\in} \: S_2 - S_1$, it declares that $e$ belongs to $S_1-S_2 \cup S_2-S_1$.
\vspace{-0.05in}
\item all $e\: \overline{\in} \: S_1-S_2$, $e\: \overline{\in} \: S_1\cap S_2$, and $e\: \overline{\in} \: S_2 - S_1$, it declares that $e$ belongs $S_1\cup S_2$.
\vspace{-0.05in}
\end{enumerate}

Note that for all these seven outcomes, the decisions of ShBF$_\text{A}$ do not suffer from false positives or false negatives.
However, decisions 4 through 6 provide slightly incomplete information and the decision 7 does not provide any information because it is already given that $e$ belongs to $S_1\cup S_2$.
We will shortly show that the probability that decision of ShBF$_\text{A}$ is one of the decisions 4 through 7 is very small, which means that with very high probability, it gives a decision with \emph{clear} meaning, and we call it a \textit{clear answer}.

\presub
\subsection{ShBF$_\textbf{A}$ -- Updating} \postsub
Just like BF handles updates by replacing each bit by a counter, we can also extend ShBF$_\text{A}$ to handle updates by replacing each bit by a counter.
We use CShBF$_\text{A}$ to denote this counting version of ShBF$_\text{A}$.
Let $C$ denote the array of $m$ counters.
To insert an element $e$, after querying $T_1$ and $T_2$ and determining whether $o(e)=0$, $o_1(e)$, or $o_2(e)$, instead of setting $k$ bits to 1, we increment each of the corresponding $k$ counters by 1; that is, we increment the $k$ counters $C[h_1(e)\%m+o(e)], \cdots, C[h_k(e)\%m+o(e)]$ by 1.
To delete an element $e$, after querying $T_1$ and $T_2$ and determining whether $o(e)=0$, $o_1(e)$, or $o_2(e)$, we decrement $C[h_i(e)\%m+o(e)]$ by 1 for all $1\leqslant i \leqslant k$.
To have the benefits of both fast query processing and small memory consumption, we maintain both ShBF$_\text{A}$ and CShBF$_\text{A}$, but store array $B$ in fast SRAM and array $C$ in slow DRAM.
After each update, we synchronize array $C$ with array $B$.

\presub
\subsection{ShBF$_\textbf{A}$ -- Analysis} \postsub
Recall from Section \ref{subsec:ShBFAQueryPhase} that ShBF$_\text{A}$ may report seven different outcomes.
Next, we calculate the probability of each outcome.
Let $P_i$ denote the probability of the $i^{\text{th}}$ outcome.
{\color{greener}
Before proceeding, we show that $h_i(.)+o(.)$ and $h_j(.)+o(.)$, when $i\neq j$, are independent of each other.
For this we show that given two random variables X and Y and a number $z\in R^+$, where $R^+$ is the set of positive real numbers, if X and Y are independent, then X+$z$ and Y+$z$ are independent.
As X and Y are independent, for any $x\in R$ and $y\in R$, we have
\vspace{-0.07in}
\begin{equation}
P(X\leqslant x, Y\leqslant y)=P(X\leqslant x)*P(Y\leqslant y)
\vspace{-0.07in}
\end{equation}
Adding $z$ to both sides of all inequality signs in $P(X\leqslant x, Y\leqslant y)$, we get
\vspace{-0.07in}
\begin{equation}
\begin{aligned}
P(X+z & \leqslant x+z, Y+z\leqslant y+z) \\
& =P(X \leqslant x, Y \leqslant y) \\
& = P(X \leqslant x) * P(Y \leqslant y) \\
& = P(X+z \leqslant x+z) * P(Y=z \leqslant y+z)
\end{aligned}
\vspace{-0.07in}
\end{equation}
Therefore, $X+z$ and $Y+z$ are independent.
}

Let $n'$ be the number of distinct elements in $S_1\cup S_2$, and let $k$ be the number of hash functions.
After inserting all $n'$ elements into ShBF$_\text{A}$, the probability $p'$ that any given bit is still 0 is given by the following equation.
\vspace{-0.07in}
\begin{equation}
p'=\left( 1-\dfrac{1}{m}\right) ^{kn'}
\vspace{-0.07in}
\end{equation}
This is similar to one minus the false positive probability of a standard BF.
When $k=\ln2 \frac{m}{n'}$, $p'\approx 0.5$.

Note that the probabilities for outcomes 1, 2, and 3 are the same.
Similarly, the probabilities for outcomes 4, 5, and 6 are also the same.
Following equations state the expressions for these probabilities.
\vspace{-0.05in}
\begin{equation}
\begin{aligned}
P_1 & =P_2=P_3=(1-0.5^k)^2     \\
P_4 & =P_5=P_6=0.5^k*(1-0.5^k)  \\
P_7 & =(0.5^k)^2        \\
\end{aligned}
\label{equ:p3}
\end{equation}
\vspace{-0.1in}

When the incoming element $e$ actually belongs to one of the three sets: $S_1-S_2$, $S_1\cap S_2$, and $S_2-S_1$, there is one combination each for $S_1-S_2$ and $S_2-S_1$ and two combinations for $S_1\cap S_2$.
Consequently, the total probability is $P_1+P_4*2+P_7$, which equals 1.
This validates our derivation of the expressions in Equation \ref{equ:p3}.
As an example, let $k$$=$$\frac{m}{n'}\ln{2}$$=$$10$.
Thus, $P_1$$=$$P_2$$=$$P_3$$=$$(1-0.5^{10})^2\approx 0.998$, $P_4$$=$$P_5$$=$$P_6$$=$$0.5^{10}*(1-0.5^{10})=9.756*10^{-4}$, and $P_7=(1-0.5^{10})^2\approx 9.54*10^{-7}$.
This example shows that with probability of 0.998, ShBF$_\text{A}$ gives a \textit{clear} answer, and with probability of only $9.756*10^{-4}$, ShBF$_\text{A}$ gives an answer with incomplete information.
The probability with which it gives an answer with no information is just $9.54*10^{-7}$, which is negligibly small.

\presub
\vspace{-0.05in}
\subsection{Comparison between ShBF$_\text{A}$ with iBF} \postsub
\label{sec:iBFcompare}
For association queries, a straightfoward solution is to build one individual BF (iBF) for each set.
Let $n_1$, $n_2$, and $n_3$ be the number of elements in $S_1$, $S_2$, and $S_1\cap S_2$, respectively. For iBF, let $m_1$ and $m_2$ be the size of the Bloom filter for $S_1$ and $S_2$, respectively.
Table \ref{table:compare:ibf} presents a comparison between ShBF$_\text{A}$ and iBF.
We observe from the table that ShBF$_\text{A}$ needs less memory, less hash computations, and less memory accesses, and has no false positives.
\textcolor{greener}{For the iBF, as we use the traffic trace that hits the two sets with the same probability, iBF is optimal when the two BFs use identical values for the optimal system parameters and have the same number of hash functions.}
Specifically, for iBF, when $m_1+m_2=(n_1+n_2)k/\ln2$, the probability of answering a \textit{clear answer} is $\frac{2}{3}(1-0.5^k)$. For ShBF$_\text{A}$, when $m=(n_1+n_2-n_3)k/\ln2$, the probability of answering a clear answer is$(1-0.5^k)^2$.

\begin{table*}[htbp]
\centering\caption{Comparison Between ShBF$_\text{A}$ and iBF.}
\begin{tabular}{|c|c|c|c|c|c|}
\hline		&	Optimal Memory&	\#hash computations	&	\#memory accesses	&	Probability of a clear answer &	false positives	\\
\hline	iBF	&	m1+m2=(n1+n2)k/ln2	&	2k	&	2k	&	$\frac{2}{3}(1-0.5^k)$	&	YES	\\
\hline	ShBF$_\text{A}$	&	m=(n1+n2-n3)k/ln2	&	k+2	&	k	&	$(1-0.5^k)^2$	&	NO	\\
\hline
\end{tabular}
\label{table:compare:ibf}
\vspace{-0.05in}
\end{table*}


%% file: multiplicityNew.tex
\presec
\vspace{-0.05in}
\section{Multiplicity Queries} \postsec
\label{sec:MultiplicityQueries}
\vspace{-0.02in}
\Comment{For multiplicity queries, the item can appear several times in a multi-set.}
In this section, we first present the construction and query phases of ShBF for multiplicity queries.
Multiplicity queries check how many times an element appears in a multi-set.
We use ShBF$_\times$ to denote the ShBF scheme for multiplicity queries.
Second, we describe the updating methods of ShBF$_\times$ .
Last, we derive the FPR and correctness rate of ShBF$_\times$.

\presub
\vspace{-0.05in}
\subsection{ShBF$_\times$ -- Construction Phase} \postsub
\vspace{-0.02in}
The construction phase of ShBF$_\times$ proceeds in three steps.
Let $h_1(.), \cdots, h_k(.)$ be $k$ independent hash functions with uniformly distributed outputs.
First, we construct an array $B$ of $m$ bits, where each bit is initialized to 0.
Second, to store the existence information of an element $e$ of multi-set $S$, we calculate $k$ hash values $h_1(e)\%m, \cdots, h_k(e)\%m$.
To calculate the auxiliary information of $e$, which in this case is the count $c(e)$ of element $e$ in $S$, we calculate offset $o(e)$ as $o(e)=c(e)-1$.
Third, we set the $k$ bits $B[h_1(e)\%m+o(e)], \cdots, B[h_k(e)\%m+o(e)]$ to 1.
%
%
To determine the value of $c(e)$ for any element $e\in S$, we store the count of each element in a hash table and use the simplest collision handling method called collision chain.

\presub
\vspace{-0.05in}
\subsection{ShBF$_\times$ -- Query Phase} \postsub
\label{sec:queryShBFtimes}
\vspace{-0.02in}
Given a query $e$, for each $1\leqslant i \leqslant k$, we first read $c$ consecutive bits $B[h_i(e)\%m]$, $B[h_i(e)\%m+1]$, $\cdots$, $B[h_i(e)\%m+c-1]$ in $\lceil\frac{c}{w}\rceil$ memory accesses, where $c$ is the maximum value of $c(e)$ for any $e\in S$.
In these $k$ arrays of $c$ consecutive bits, for each $1 \leqslant j \leqslant c$, if all the $k$ bits $B[h_1(e)\%m+j-1], \cdots, B[h_k(e)\%m + j-1]$ are 1, we list $j$ as a possible candidate of $c(e)$.
As the largest candidate of $c(e)$ is always greater than or equal to the actual value of $c(e)$, we report the largest candidate as the multiplicity of $e$ to avoid false negatives.
%
%
For the query phase, the number of memory accesses is $k\lceil\frac{c}{w}\rceil$.

\presub
\vspace{-0.05in}
\subsection{ShBF$_{\times}$ -- Updating} \postsub
\vspace{-0.07in}
\subsubsection{\textcolor{reder}{ShBF$_\times$ -- Updating with False Negatives}}
\vspace{-0.02in}
To handle element insertion and deletion, ShBF$_{\times}$ maintains its counting version denoted by CShBF$_{\times}$, which is an array $C$ that consists of $m$ counters, in addition to an array $B$ of m bits.
During the construction phase, ShBF$_\times$ increments the counter $C[h_i(e)\%m+o(e)]$ ($1\leqslant i \leqslant k$) by one every time it sets $B[h_i(e)\%m+o(e)]$ to 1.
During the update, we need to guarantee that one element with multiple multiplicities is always inserted into the filter one time.
Specifically, for every new element $e$ to insert into the multi-set $S$, ShBF$_{\times}$ first obtains its multiplicity $z$ from $B$ as explained in Section \ref{sec:queryShBFtimes}.
Second, it deletes the $z-$th multiplicity ($o(e)=z-1$) and inserts the $(z+1)-$th multiplicity ($o(e)=z$).
For this, it calculates the $k$ hash functions $h_i(e)\%m$ and decrements the $k$ counters $C[h_i(e)\%m+z-1]$  by 1 when the counters are $\geqslant1$.
Third, if any of the decremented counters becomes 0, it sets the corresponding bit in $B$ to 0.
Note that maintaining the array $C$ of counters allows us to reset the right bits in $B$ to 0.
Fourth, it increments the $k$ counters $C[h_i(e)\%m+z]$ by 1 and sets the bits $B[h_i(e)\%m+z]$ to 1.

For deleting element $e$, ShBF$_{\times}$ first obtains its multiplicity $z$ from $B$ as explained in Section \ref{sec:queryShBFtimes}. 
Second, it calculates the $k$ hash functions and decrements the counters $C[h_i(e)\%m+z-1]$  by 1.
Third, if any of the decremented counters becomes 0, it sets the corresponding bit in $B$ as 0.
Fourth, it increments the counters $C[h_i(e)\%m+z-2]$ by 1 and sets the bits $B[h_i(e)\%m+z-2]$ to 1.

Note that ShBF$_{\times}$ may introduce false negatives because before updating the multiplicity of an element, we first query its current multiplicity from $B$.
If the answer to that query is a false positive, \ie, the actual multiplicity of the element is less than the answer, ShBF$_{\times}$ will perform the second step and decrement some counters, which may cause a counter to decrement to 0.
Thus, in the third step, it will set the corresponding bit in $B$ to 0, which will cause false negatives.

\vspace{-0.05in}
\subsubsection{\textcolor{reder}{ShBF$_\times$ -- Updating without False Negatives}}
\vspace{-0.02in}
%
To \textit{eliminate false negatives}, in addition to arrays $B$ and $C$, ShBF$_\times$ maintains a hash table to store counts of each element. In the hash table, each entry has two fields: element and its counts/multiplicities.
%
{\color{reder}When inserting or deleting element $e$, ShBF$_\times$ follows four steps shown in Figure \ref{draw:shbfxupdate}.
First, we obtain $e$'s counts/multiplicities from the hash table instead of ShBF$_\times$.
Second, we delete $e$'s $z$-th multiplicity from CShBF$_\times$.
Third, if a counter in CShBF$_\times$ decreases to 0, we set the corresponding bit in ShBF$_\times$ to 0.
Fourth, when inserting/deleting $e$, we insert the $(z-1)-$th/$(z+1)-$th multiplicity into ShBF$_\times$.
}
\vspace{-0.15in}
\begin{figure}[htbp]
\centering
\includegraphics[width=0.5\textwidth]{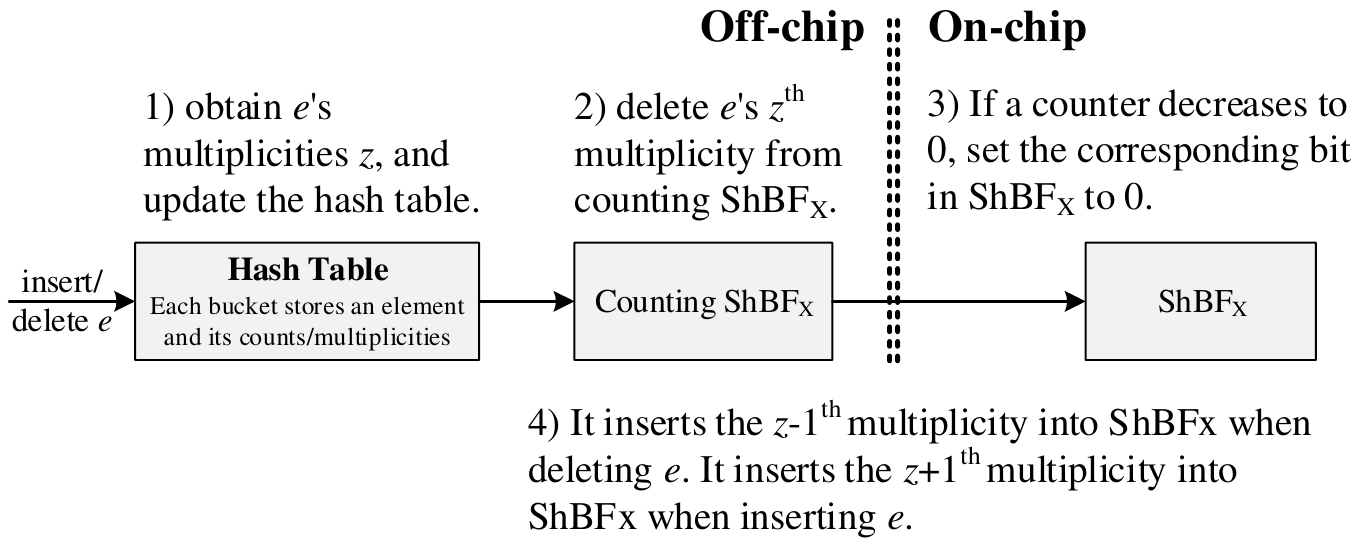}
\vspace{-0.05in}
\caption{{\color{reder}The update process of ShBF$_\times$.}}
\vspace{-0.05in}
\label{draw:shbfxupdate}
\end{figure}

%
Note that although the counter array $C$ and the hash table are much larger than the bit array $B$, we store $B$ in SRAM for processing multiplicity queries and store $C$ and the hash table in DRAM for handling updates.

\presub
\subsection{ShBF$_\times$ -- Analysis} \postsub
%
%
For multiplicity queries a false positive is defined as reporting the multiplicity of an element that is larger than its actual multiplicity.
For any element $e$ belonging to multi-set $S_m$, ShBF$_{\times}$ only sets $k$ bits in $B$ to 1 regardless of how many times it appears in $S_m$.
This is because every time information about $e$ is updated, ShBF$_{\times}$ removes the existing multiplicity information of the element before adding the new information.
Let the total number of distinct elements in set $S_m$ be $n$.
The probability that an element is reported to be present $j$ times is given by the following equation.
\begin{equation}
f_0 \approx \left(1- e^{-\frac{kn}{m}} \right)^{k}
\end{equation}
%
%

We define a metric called \textit{correctness rate}, which is the probability that an element that is present $j$ times in a multi-set is correctly reported to be present $j$ times.
{\color{greener}
When querying an element not belonging to the set, the correctness rate $CR$ is given by the following equation.
\begin{equation}
CR=(1-f_0)^c
\label{equ:correct:ShBF_X}
\end{equation}
When querying an element with multiplicity $j$ $(1\leqslant j\leqslant c)$ in the set, the correctness rate $CR'$ is given by the following equation.
\begin{equation}
CR'=(1-f_0)^{j-1}
\end{equation}
Note the right hand side of the expression for $CR'$ is not multiplied with $f_0$ because when $e$ has $j$ multiplicities, all positions $h_i(e)+j$, where $1\leqslant i\leqslant k$, must be 1.
}


{\color{bluer}
\presub \subsection{Shifting Count-min Sketch} \postsub
\label{sec:CMsketch}
Besides Spectral BF, count-min sketch (CM sketch) can also be used to record and report the number of multiplicities of each element \cite{CMsketch}.
As shown in Figure \ref{draw:CMsketch}, a CM sketch consists of $d$ vectors, and each vector has $r$ counters.
Each vector $v_i$ $(1\leqslant i\leqslant d)$ corresponds to a hash function $h_i(.)$.
When inserting an element $e$, the CM sketch increments the counters $v_1[h_1(e)]...v_d[h_d(e)]$ by 1.
When querying an element $e$, the CM sketch reports the minimum value of $v_1[h_1(e)]...v_d[h_d(e)]$.
CM sketch is simple and easy to implement, but is not memory efficient, as the minimal unit is a counter instead of a bit.

\begin{figure}[htbp]
\centering
\hspace{-0.11in}
\subfigure[CM sketch]{
{\includegraphics[width=0.48\columnwidth]{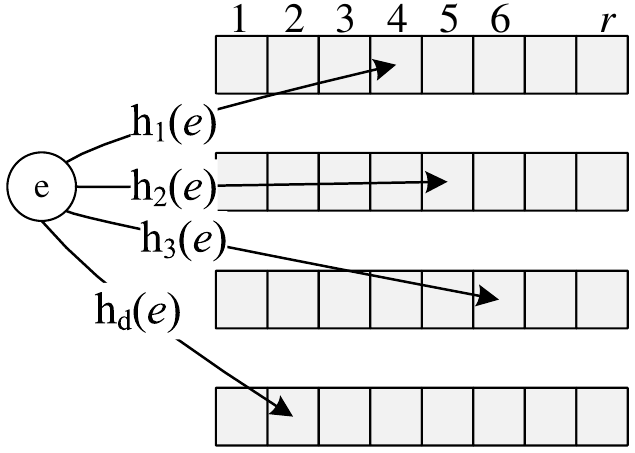}}
\label{draw:CMsketch}}
\hspace{-0.11in}
\subfigure[SCM sketch]{
{\includegraphics[width=0.48\columnwidth]{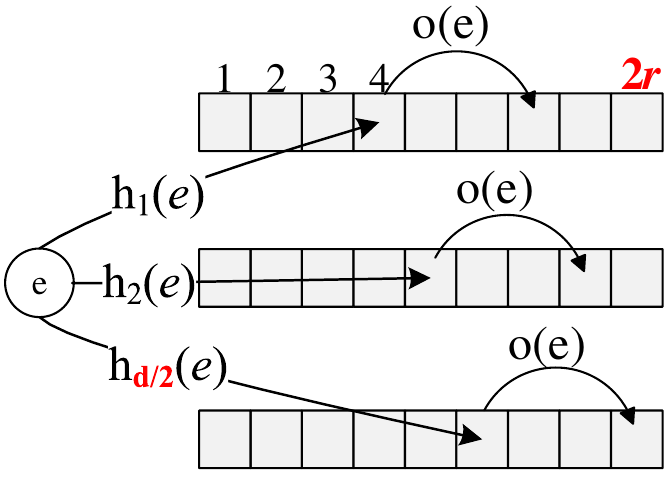}}
\label{draw:SCMsketch}}
\prefigcaption
\caption{FPR vs. $\overline{w}$.}
\label{draw:2sketch}
\postfig
\end{figure}


One query on CM sketch needs $d$ hash computations and memory accesses when the length of all counters is smaller than a machine word.
In such a case, we can use our shifting framework to halve the number of memory accesses and hash functions.
Figure \ref{draw:CMsketch} shows our shifting version of CM sketch, called shifting count-min (SCM) sketch.
SCM sketch consists of $d/2$ vectors $v_i$ $(1\leqslant i\leqslant d/2)$, where each vector has $2r$ counters.
Each vector $v_i$ corresponds to a hash function $h_i(.)$.  When inserting an element $e$, first, the SCM sketch increments the counters $v_1[h_1(e)]...v_{d/2}[h_{d/2}(e)]$ by 1.
Second, it increments the counters $v_1[h_1(e)+o(e)]...v_{d/2}[h_{d/2}(e)+o(e)]$ by 1.
When querying an element $e$, the SCM sketch reports the minimal value of $v_1[h_1(e)]...v_{d/2}[h_{d/2}(e)]$ and $v_1[h_1(e)+o(e)]...v_{d/2}[h_{d/2}(e)+o(e)]$.
To read $v_i[h_i(e)]$ and $v_i[h_i(e)+o(e)]$ in one memory accesses, we set $o(e)=h_{d/2+1} \%(\overline{w}-1)+1$, where $\overline{w}\leqslant (w-7)/r$ and $w$ is the number of bits in a machine word.
%
%
%
}

%% file: evaluation.tex
\begin{figure*}[htbp]
\centering
\hspace{-0.11in}
\subfigure[Changing $n$]{
{\includegraphics[width=0.32\textwidth]{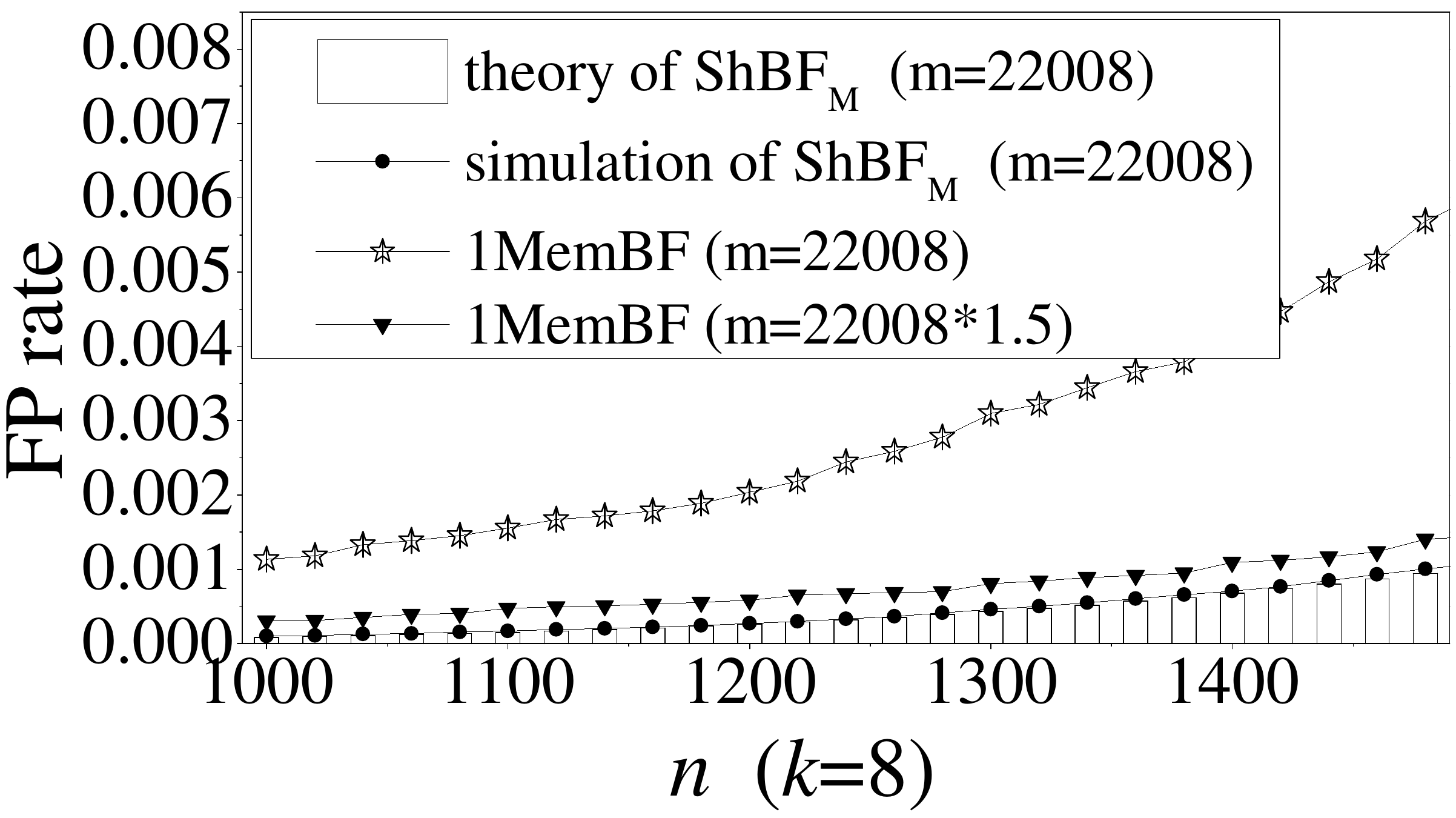}}
\label{eva:Single:FP:n}}
\subfigure[Changing $k$]{
{\includegraphics[width=0.32\textwidth]{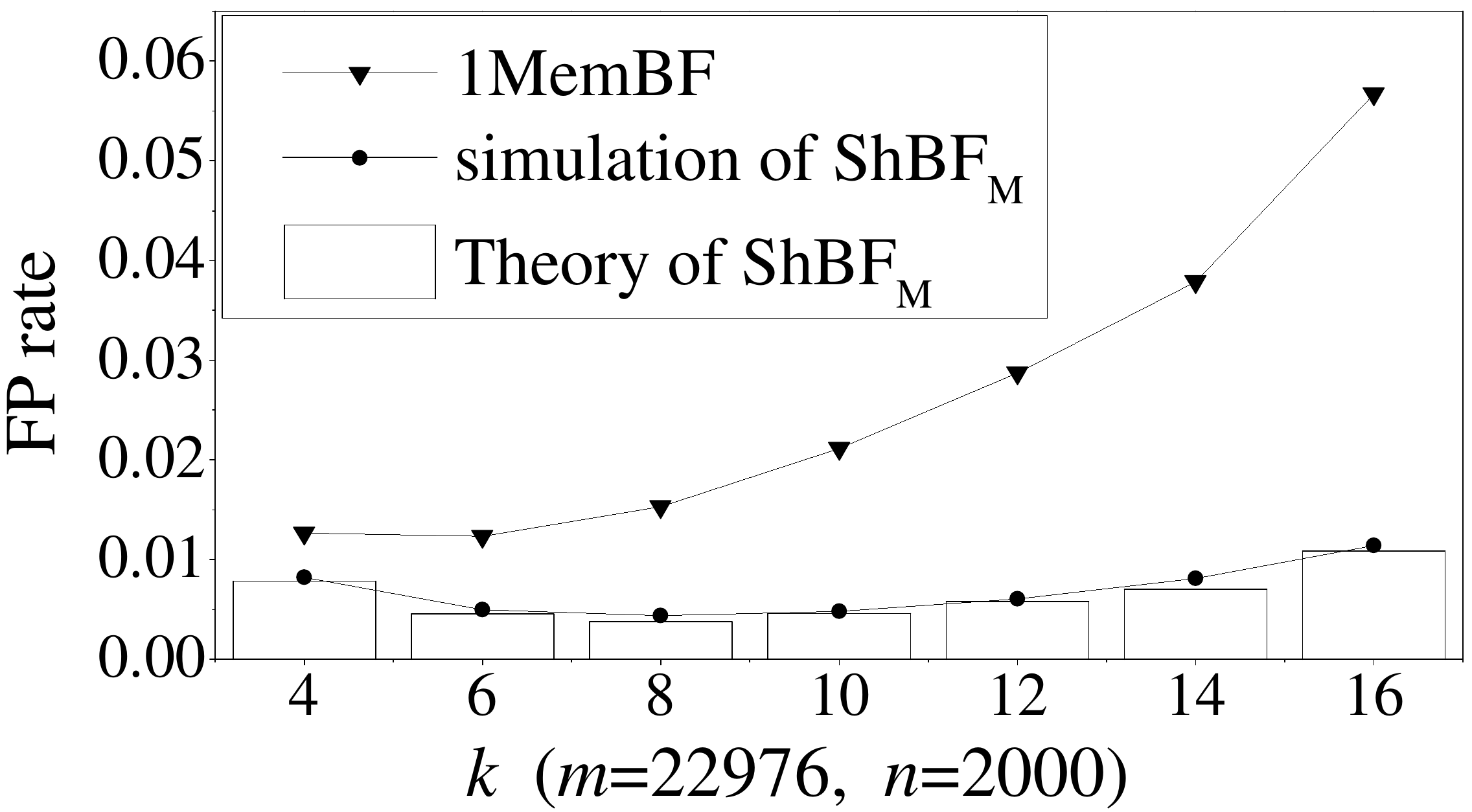}}
\label{eva:Single:FP:k}}
\subfigure[Changing $m$]{
{\includegraphics[width=0.32\textwidth]{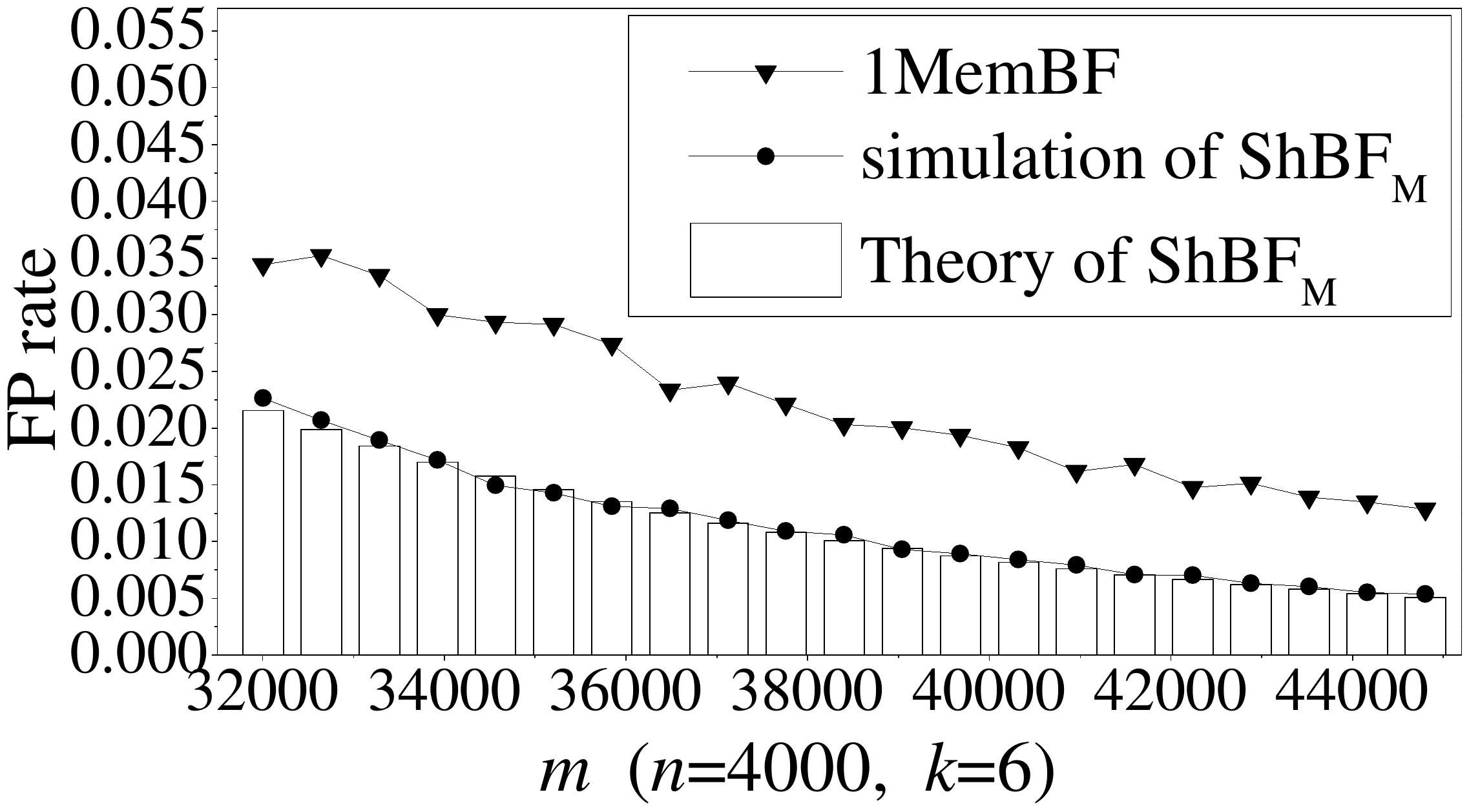}}
\label{eva:Single:FP:m}}
\vspace{-0.1in}
\caption{{\color{greener}Comparison false positive rates of ShBF$_\text{M}$ and 1MemBF.}}
\hspace{-0.11in}
\subfigure[Changing $n$]{
{\includegraphics[width=0.32\textwidth]{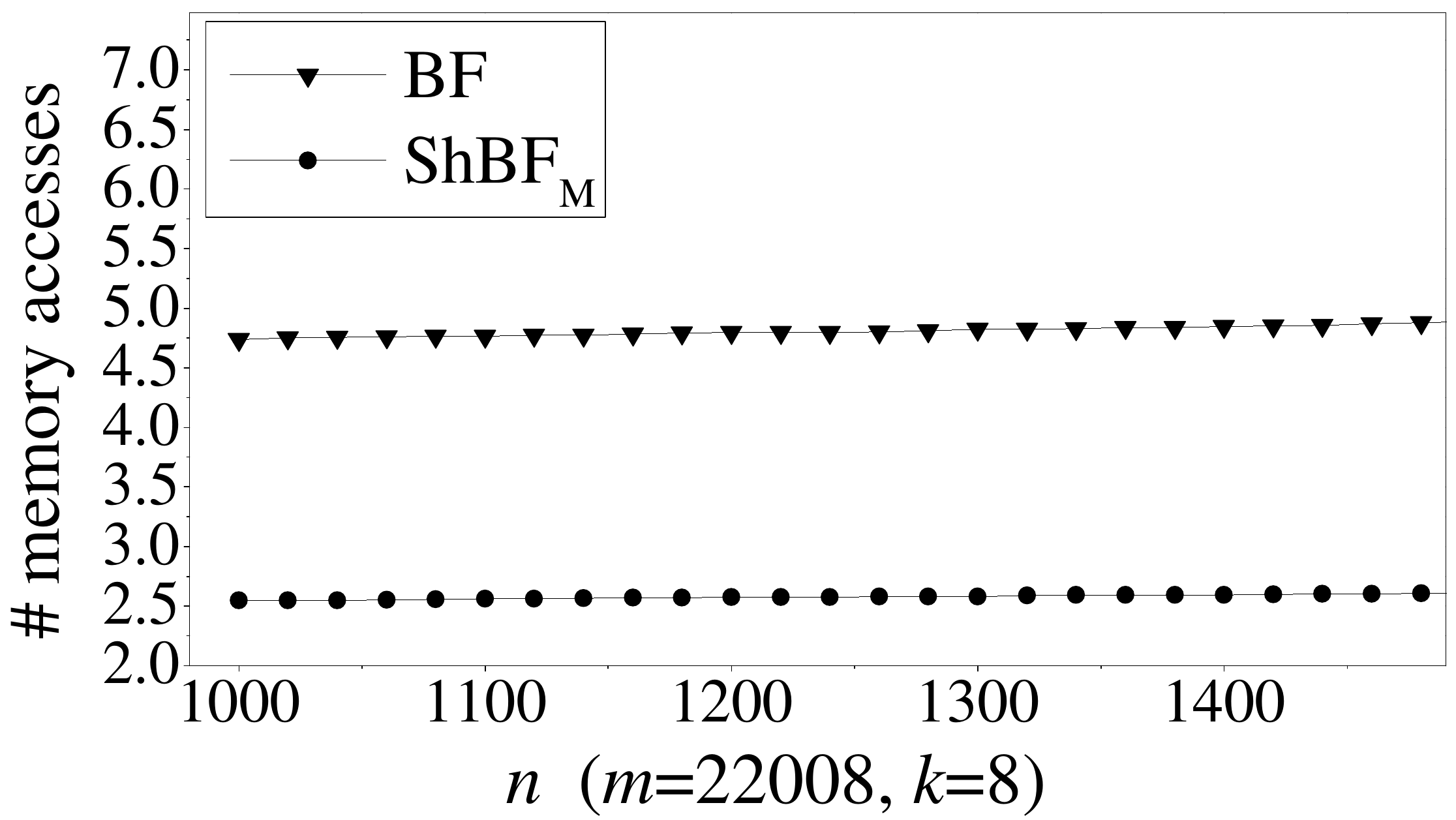}}
\label{eva:Single:Mem:n}}
\subfigure[Changing $k$]{
{\includegraphics[width=0.31\textwidth]{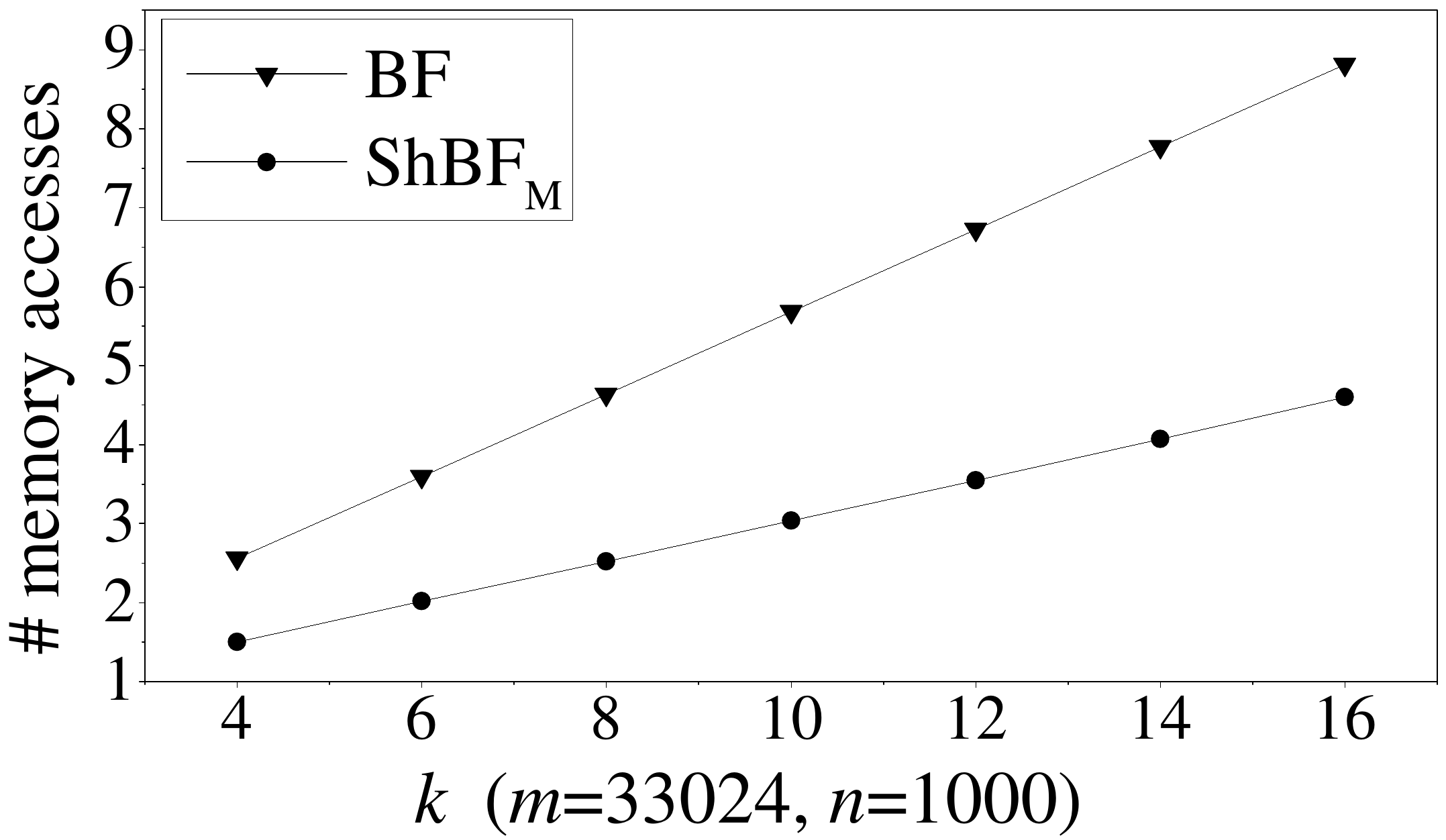}}
\label{eva:Single:Mem:k}}
\subfigure[Changing $m$]{
{\includegraphics[width=0.32\textwidth]{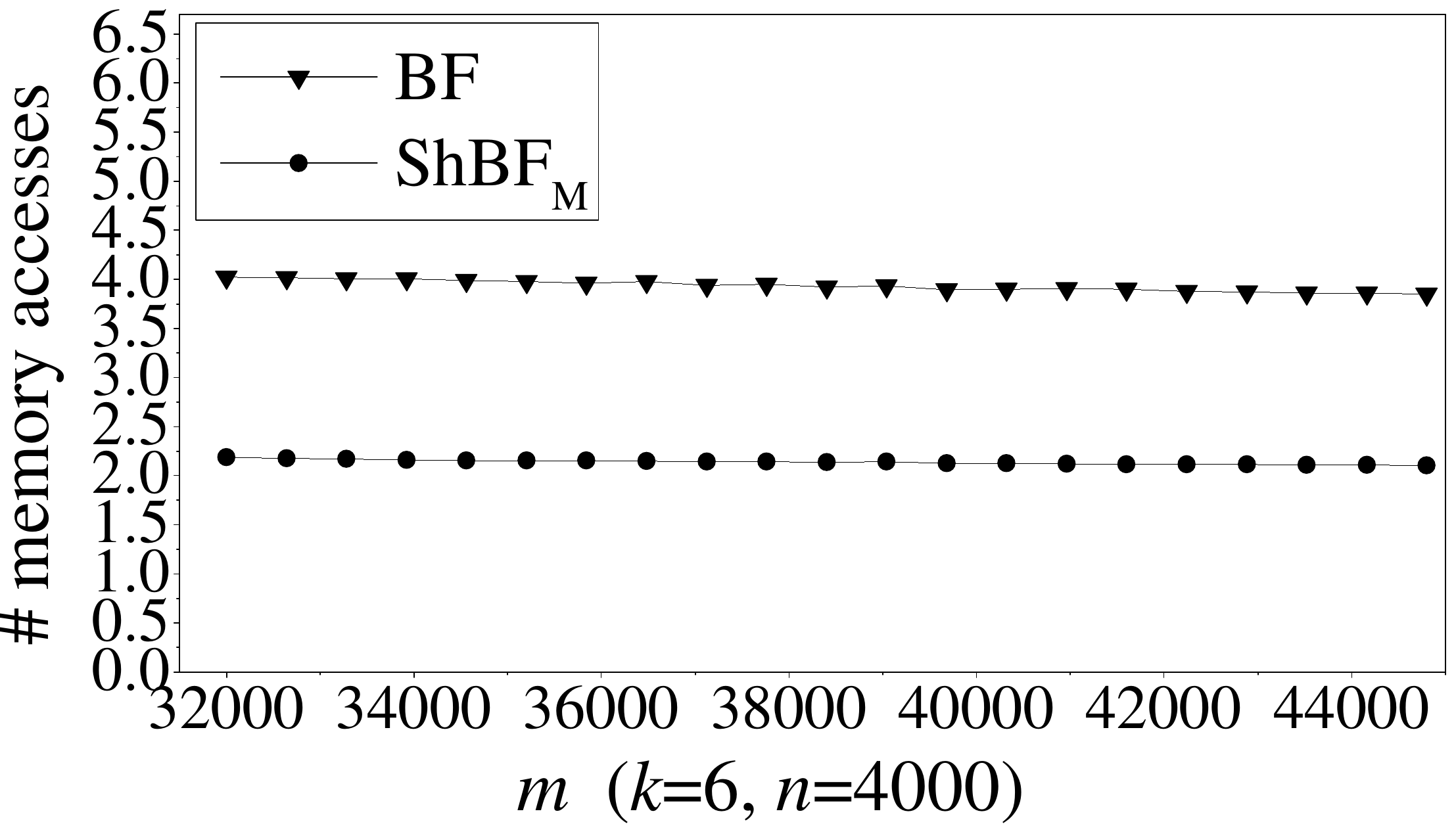}}
\label{eva:Single:Mem:m}}
\vspace{-0.1in}
\caption{{\color{greener}Comparison of number of memory accesses per query of ShBF$_\text{M}$ and BF.}}
\hspace{-0.11in}
\subfigure[Changing $n$]{
{\includegraphics[width=0.32\textwidth]{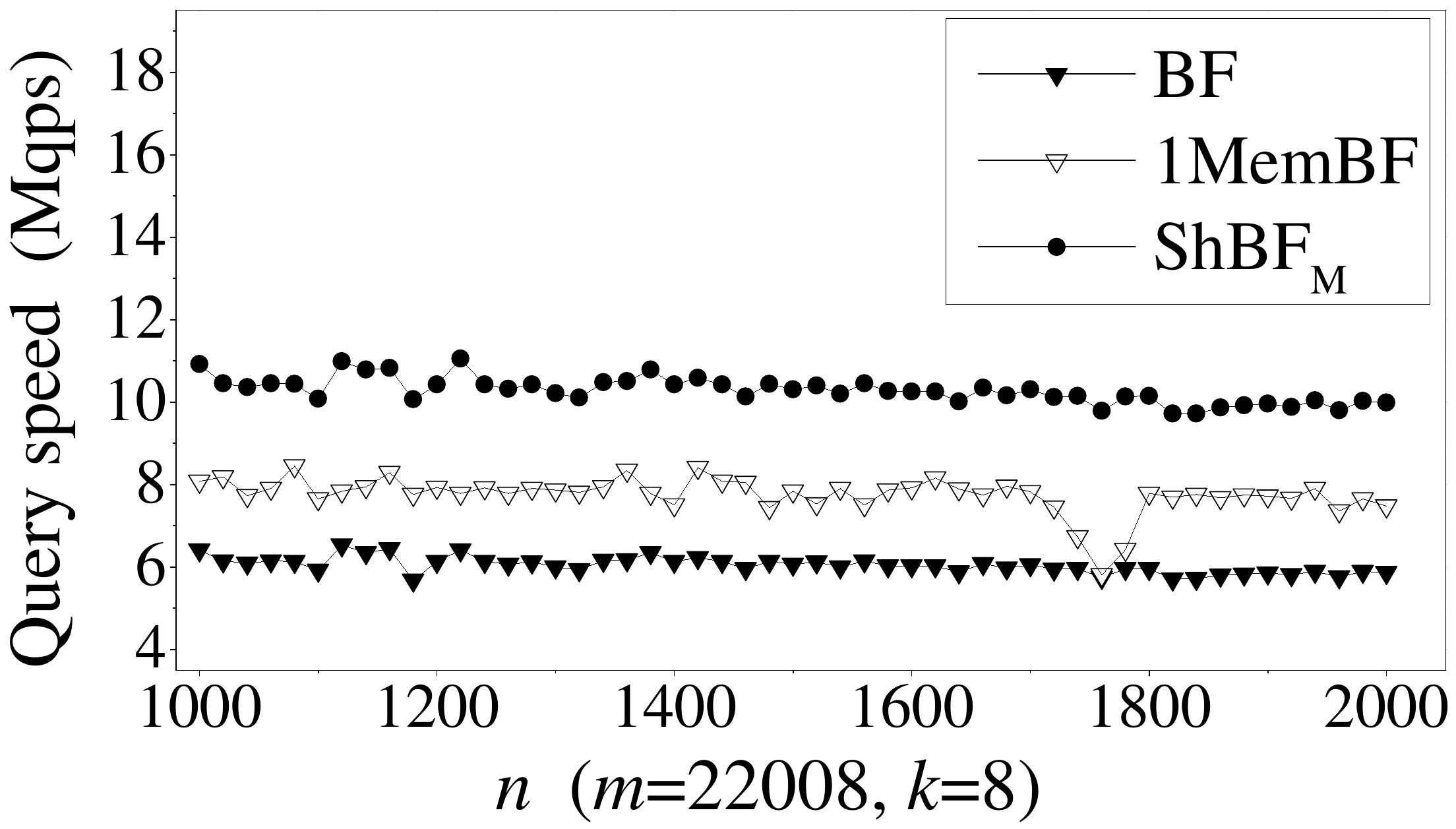}}
\label{eva:Single:Qps:n}}
\subfigure[Changing $k$]{
{\includegraphics[width=0.32\textwidth]{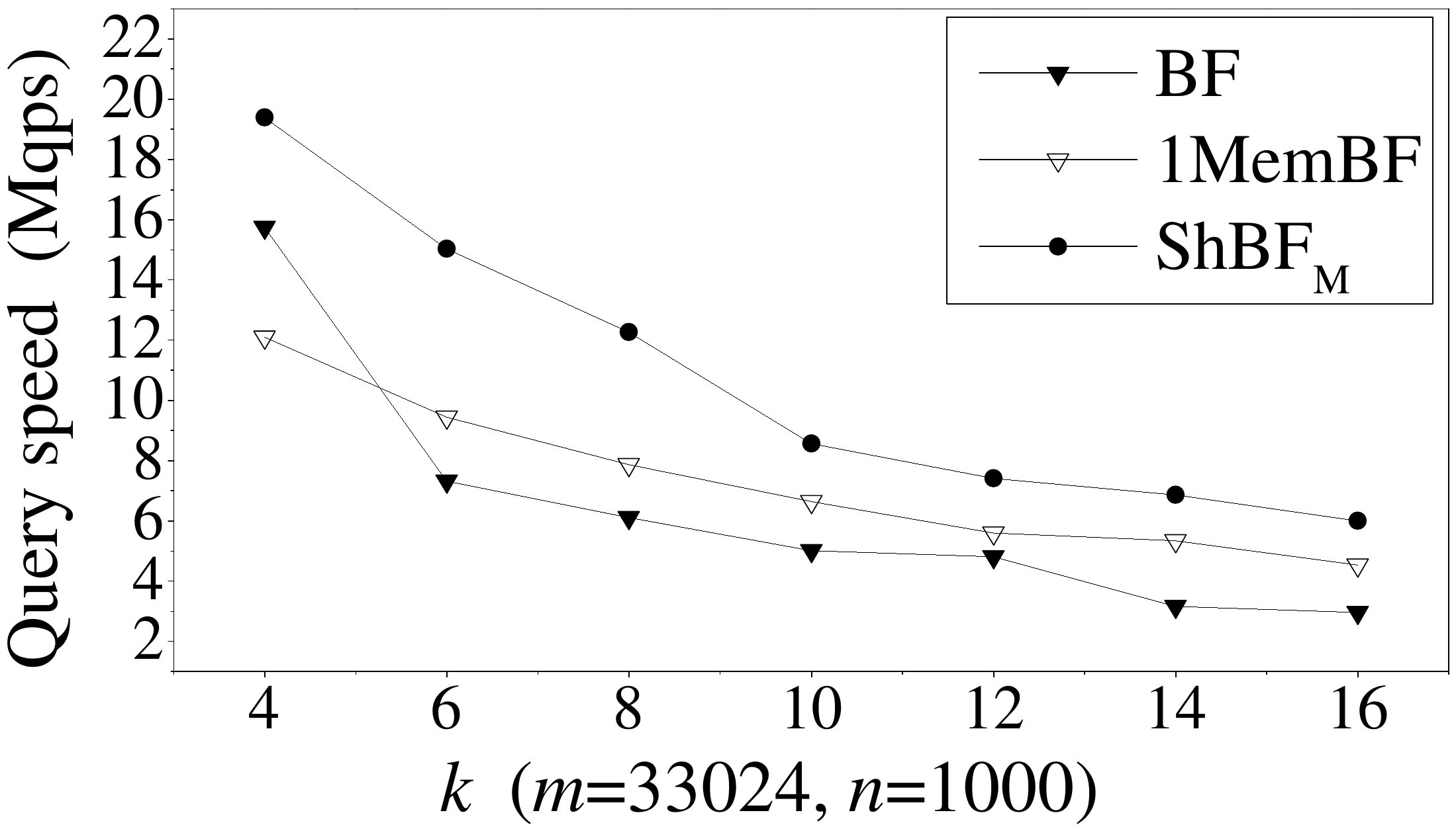}}
\label{eva:Single:Qps:k}}
\subfigure[Changing $m$]{
{\includegraphics[width=0.32\textwidth]{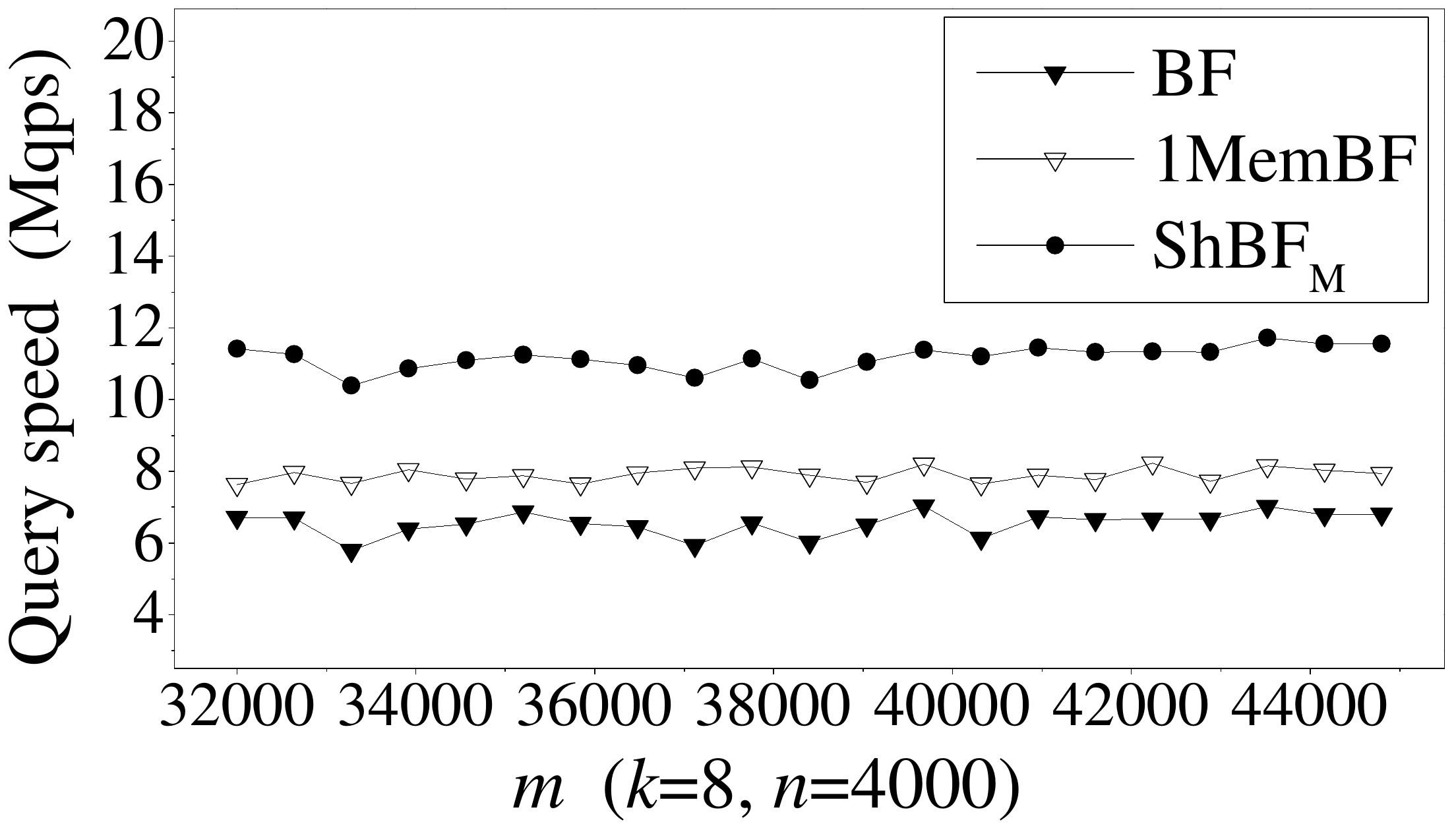}}
\label{eva:Single:Qps:m}}
\prefigcaption
\vspace{-0.1in}
\caption{{\color{greener}Comparison of query processing speeds of ShBF$_\text{M}$, BF, and 1MemBF.}}
\hspace{-0.11in}
\subfigure[Prob. a clear answer]{
{\includegraphics[width=0.32\textwidth]{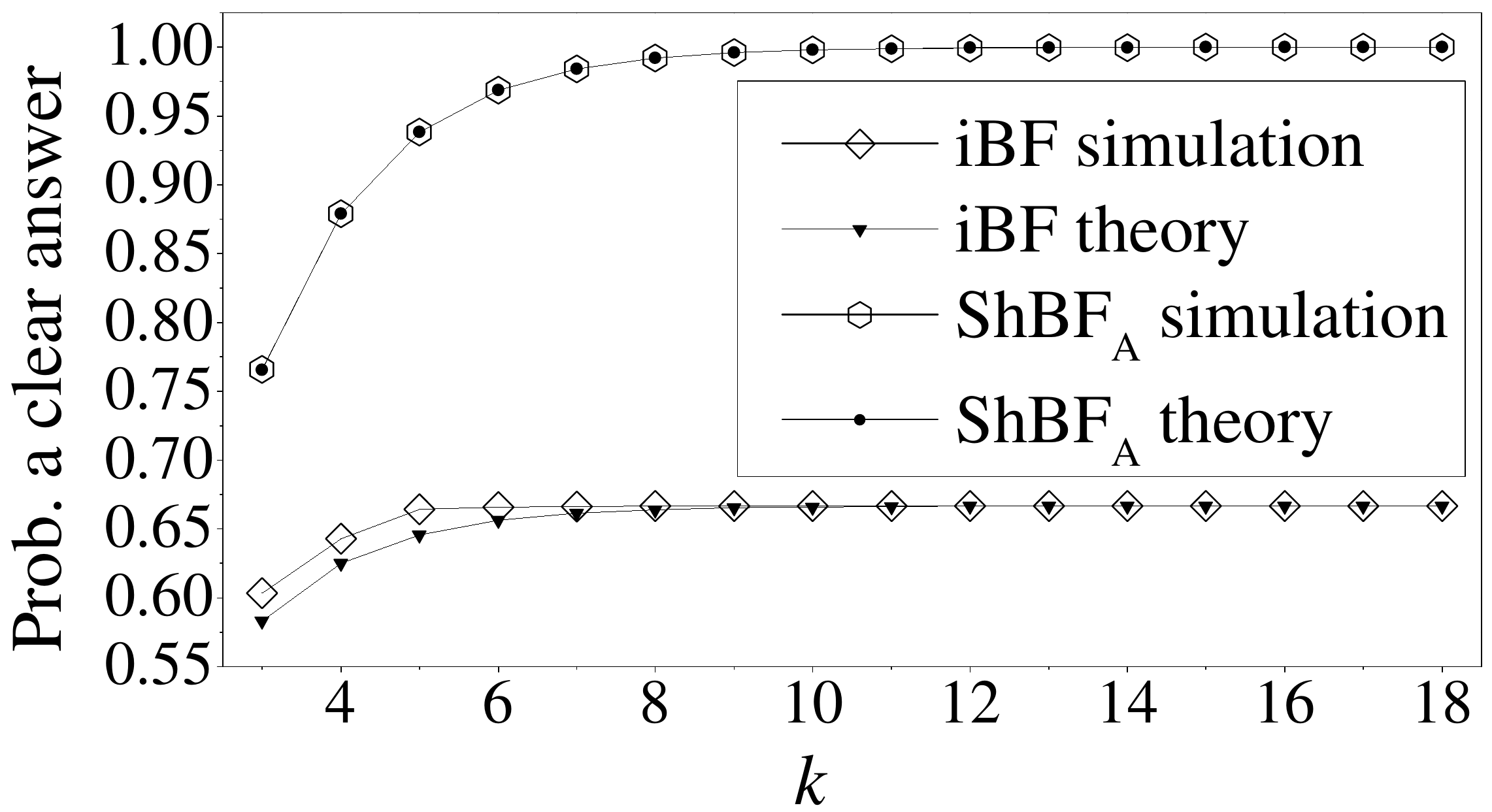}}
\label{eva:twoset:FP:k}}
\subfigure[\# memory accesses per query]{
{\includegraphics[width=0.32\textwidth]{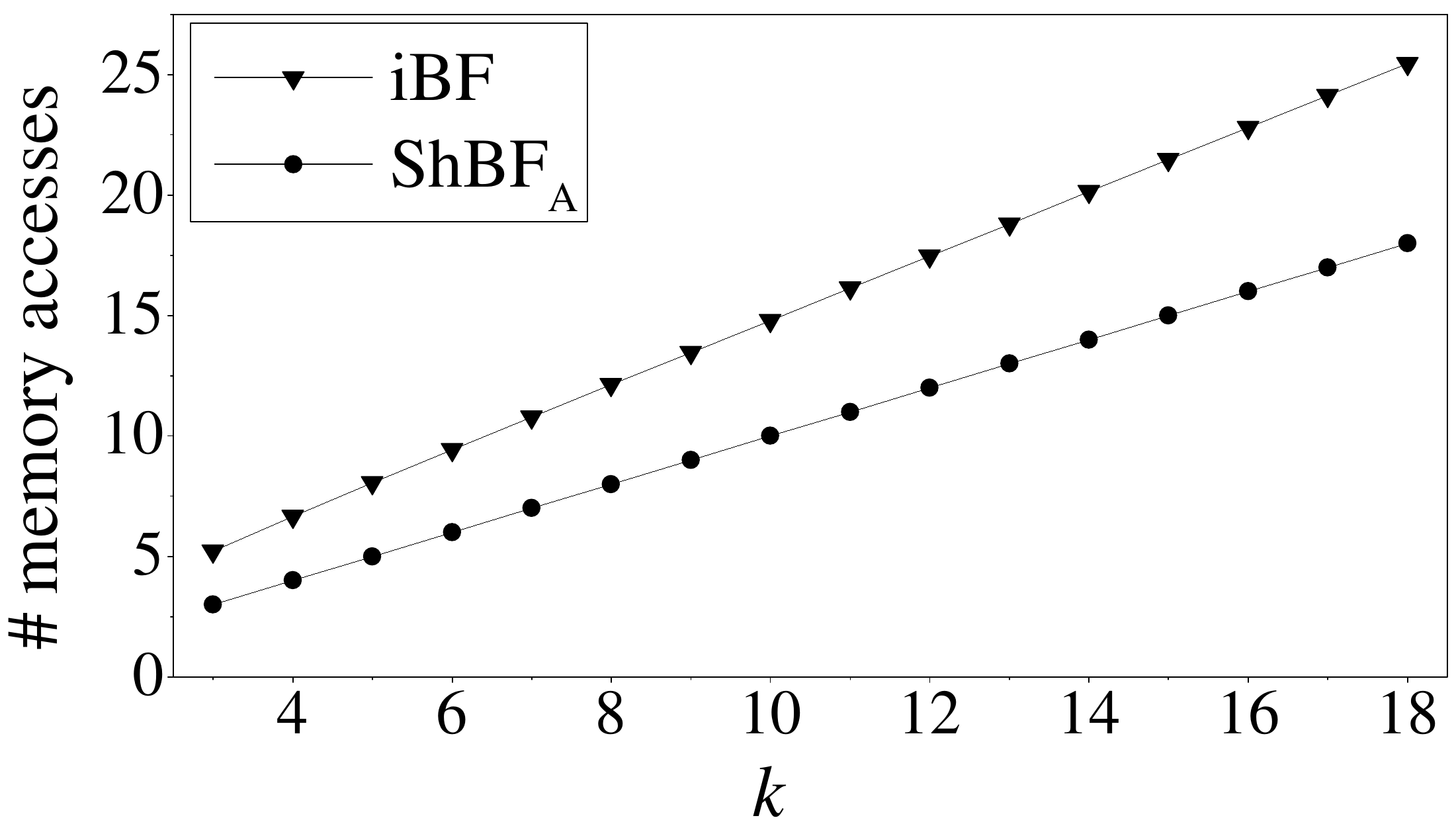}}
\label{eva:twoset:Mem:k}}
\subfigure[Query processing speed]{
{\includegraphics[width=0.32\textwidth]{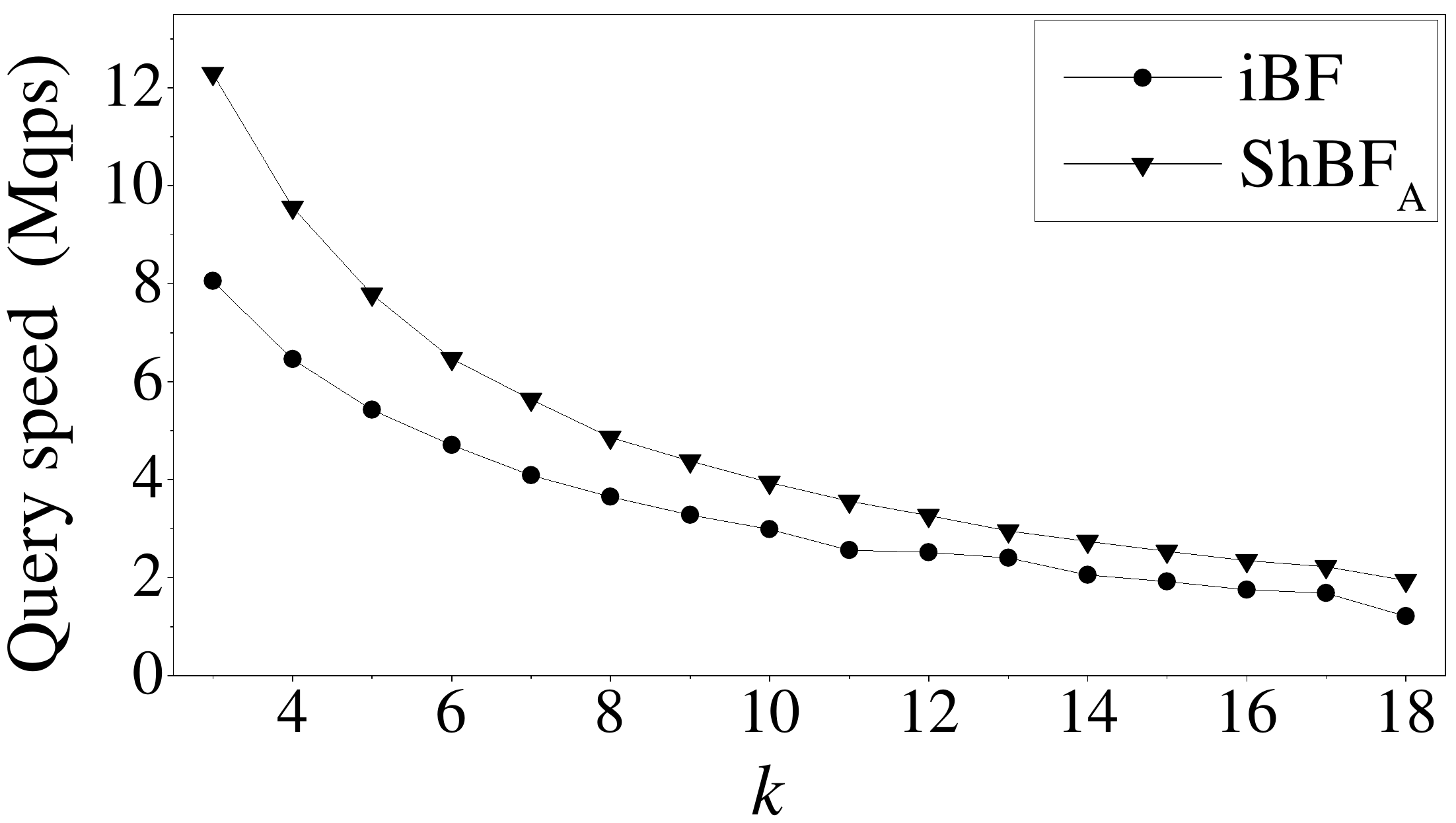}}
\label{eva:twoset:Qps:k}}
\prefigcaption
\vspace{-0.1in}
\caption{{\color{reder}Comparison of ShBF$_\text{A}$ and iBF.}}
\hspace{-0.11in}
\subfigure[Correctness rate (CR)]{
{\includegraphics[width=0.32\textwidth]{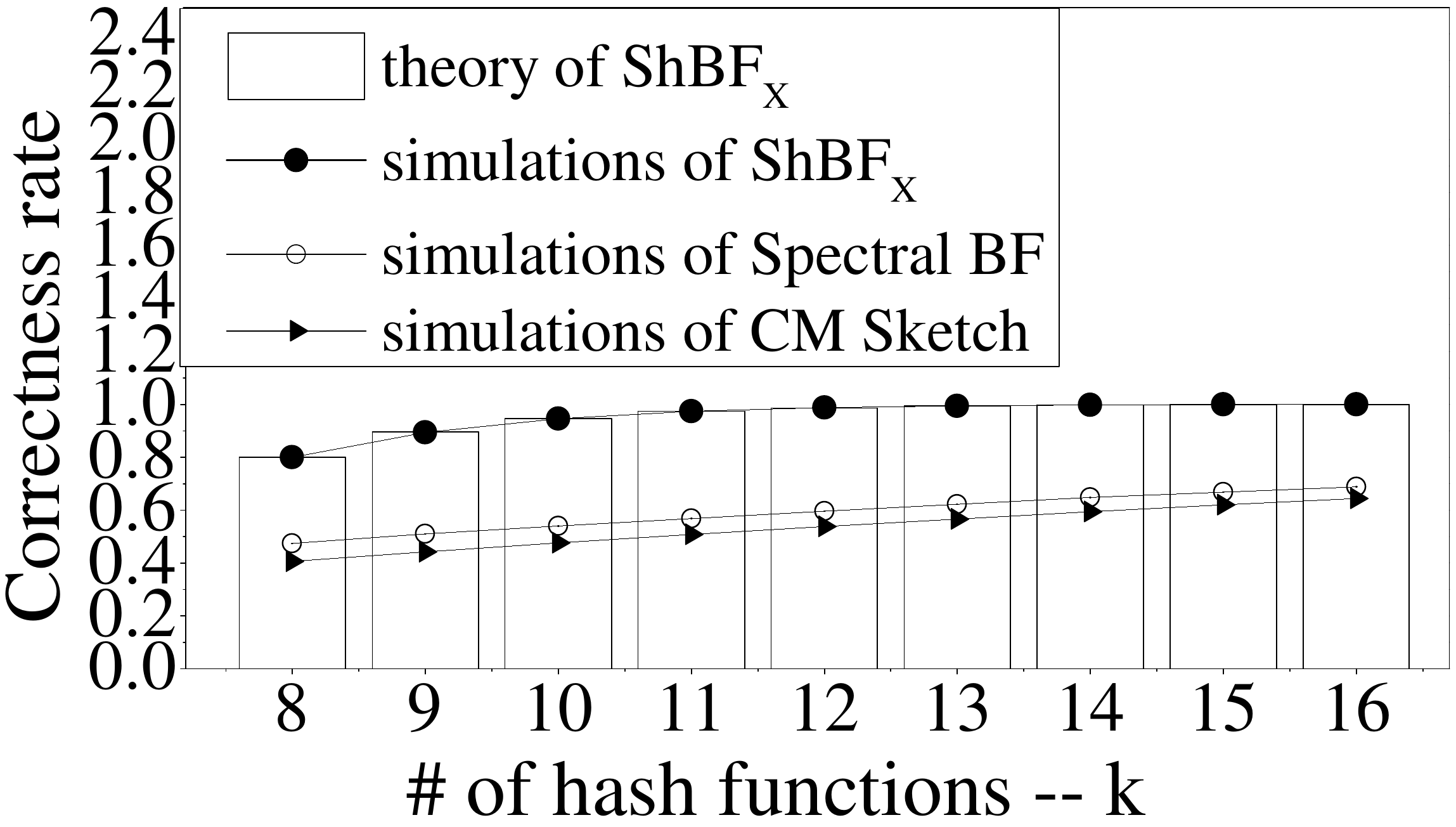}}
\label{eva:ShBF_C:FP}}
\subfigure[Memory accesses]{
{\includegraphics[width=0.32\textwidth]{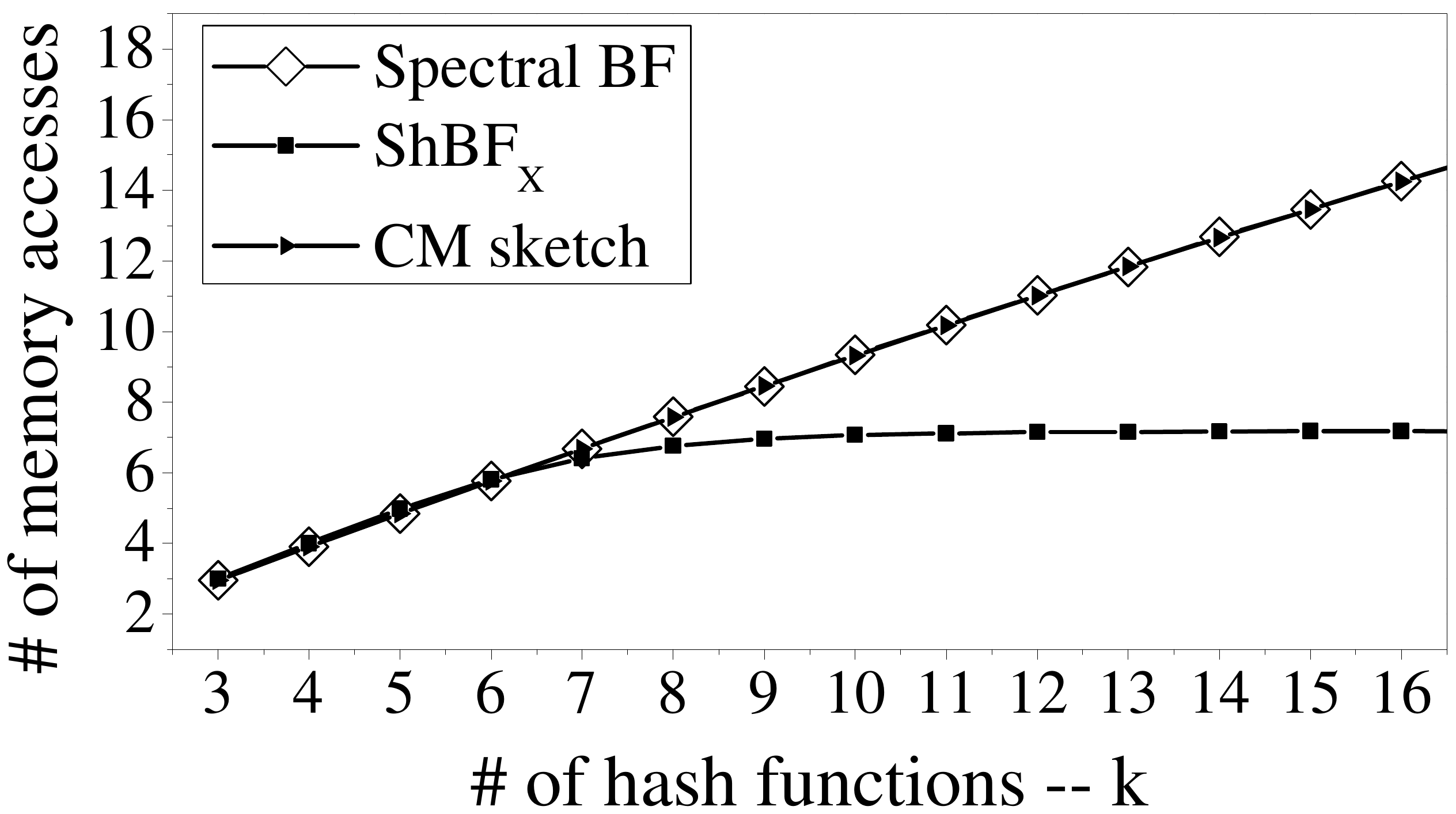}}
\label{eva:ShBF_C:MemAcc}}
\subfigure[Query processing speed]{
{\includegraphics[width=0.32\textwidth]{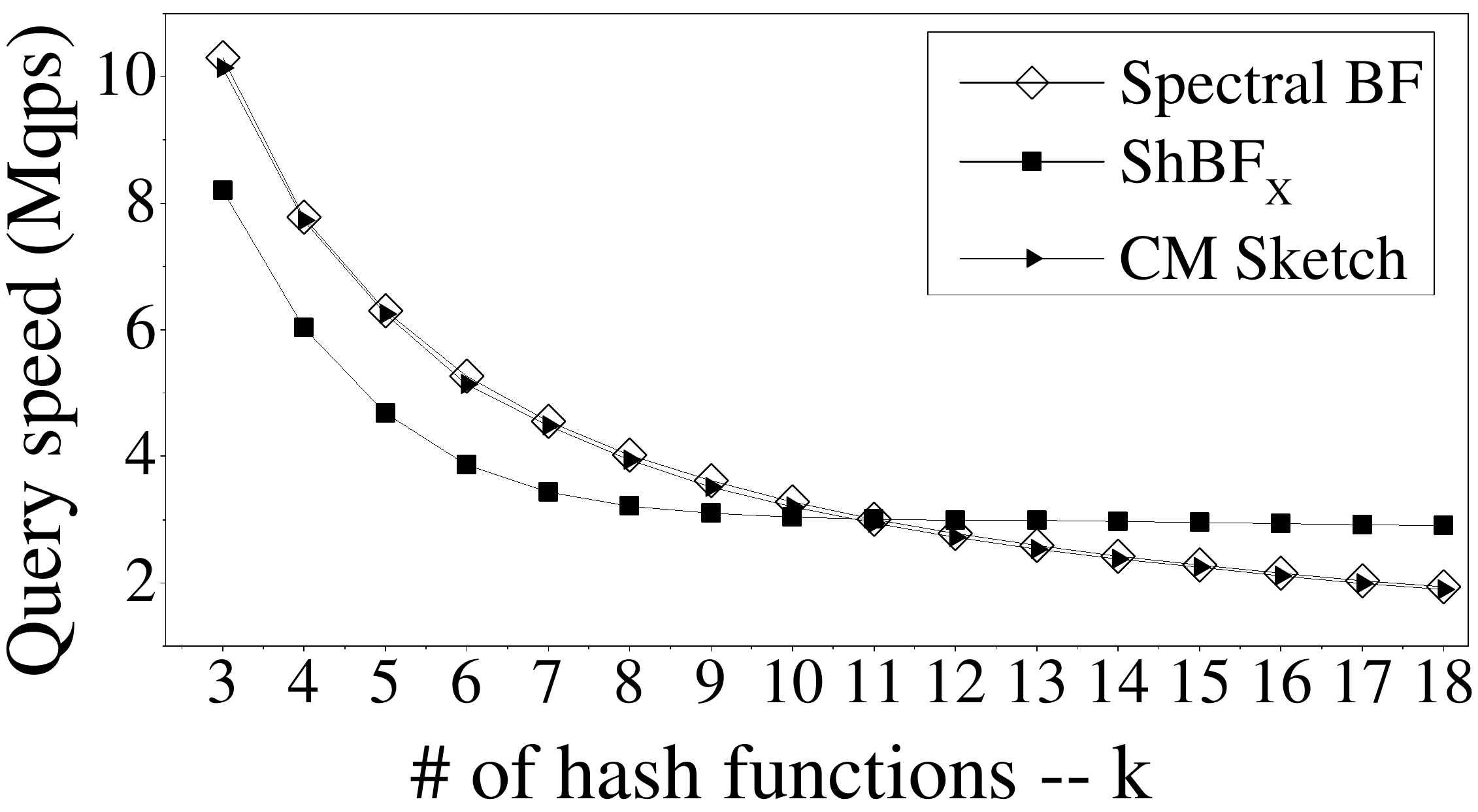}}
\label{eva:ShBF_C:Qps}}
\prefigcaption
\caption{{\color{bluer}Comparison of ShBF$_{\times}$, Spectral BF, and CM Sketch.}}
\end{figure*}

\presec
\section{Performance Evaluation}
\label{sec:evaluation}



In this section, we conduct experiments to evaluate our ShBF schemes and side-by-side comparison with state-of-the-art solutions for the three types of set queries.
%

\presub
\subsection{Experimental Setup}
\vspace{0.05in}
We give a brief overview of the data we have used for evaluation and describe our experimental setup.

\textbf{Data set:}
We evaluate the performance of ShBF and state-of-the-art solutions using real-world network traces.
Specifically, we deployed our traffic capturing system on a 10Gbps link of a backbone router.
To reduce the processing load, our traffic capturing system consists of two parallel sub-systems each of which is equipped with a 10G network card and uses \textit{netmap} to capture packets.
Due to high link speed, capturing entire traffic was infeasible because our device could not access/write to memory at such high speed.
Thus, we only captured 5-tuple flow ID of each packet, which consists of source IP, source port, destination IP, destination port, and protocol type.
We stored each 5-tuple flow ID as a 13-byte string, which is used as an element of a set during evaluation.
We collected a total of 10 million 5-tuple flow IDs, out of which 8 million flow IDs are distinct.
\textcolor{bluer}{To further evaluate the accuracy of our proposed schemes, we also generated and used synthetic data sets.
}

\textbf{Hash functions:}
We collected several hash functions from open source web site \cite{openhashWebsite} and tested them for randomness.
Our criteria for testing randomness is that the probability of seeing 1 at any bit location in the hashed value should be 0.5.
To test the randomness of each hash function, we first used that hash function to compute the hash value of the 8 million unique elements in our data set.
Then, for each bit location, we calculated the fraction of times 1 appeared in the hash values to empirically calculate the probability of seeing 1 at that bit location.
Out of all hash functions, 18 hash functions passed our randomness test, which we used for evaluation of ShBF and state-of-the-art solutions.

\textbf{Implementation:}
We implemented our query processing schemes in C++ using Visual C++ 2012 platform.
To compute average query processing speeds, we repeat our experiments 1000 times and take the average.
Furthermore, we conducted all our experiments for 20 different sets of parameters.
As the results across different parameter sets follow same trends, we will report results for one parameter set only for each of the three types of queries.

\textbf{Computing platform:}
We did all our experiments on a standard off the shelf desktop computer equipped with an Intel(R) Core i7-3520 CPU @2.90GHz and 8GB RAM running Windows 7. 

\presec
\subsection{ShBF$_{\text{M}}$ -- Evaluation} \postsub
In this section, we first validate the false positive rate of ShBF$_{\text{M}}$ calculated in Equation \eqref{equ:pfpShBFM} using our experimental results.
Then we compare ShBF$_{\text{M}}$ with BF and 1MemBF \cite{lessHashBF}, which represents the prior scheme for answering membership queries, in terms of FPR, the number of memory accesses, and query processing speed.

\vspace{-0.05in}
\presub
\subsubsection{ShBF$_{\text{M}}$ -- False Positive Rate}
\emph{Our experimental results show that the FPR of ShBF$_{\text{M}}$ calculated in Equation \eqref{equ:pfpShBFM} matches with the FPR calculated experimentally.}
For the experiments reported in this section, we set $k=8$, $m=22008$, $\overline{w}=57$, and vary $n$ from 1000 to 1500.
We first insert 1000 elements into ShBF$_{M}$ and then repeatedly insert 20 elements until the total number of elements inserted into ShBF$_{M}$ became 1500.
On inserting each set of 20 elements, we generated membership queries for $7,000,000$ elements whose information was not inserted into ShBF$_{M}$ and calculated the false positive rate.
Figure \ref{eva:Single:FP:n} shows the false positive rate of ShBF$_{\text{M}}$ calculated through these simulations as well as through Equation \eqref{equ:pfpShBFM}.
The bars in Figure \ref{eva:Single:FP:n} represent the theoretically calculated FPR, whereas the lines represent the FPR observed in our experiments.


\emph{Our results show that the relative error between the FPRs of ShBF$_{M}$ calculated using simulation and theory is less than 3\%}, which is practically acceptable.
Relative error is defined as $|FPR_s-FPR_t|/FPR_t$, where $FPR_s$ is the false positive rate calculated using simulation and $FPR_t$ is the false positive rates calculated using theory. 
The relative error of 3\% for ShBF$_{\text{M}}$ is the same as relative error for BF calculated using simulation and the theory developed by Bloom \etal \cite{BF1970}.
Using same parameters, the FPR of 1MemBF is over $5\sim 10$ times that of ShBF$_\text{M}$.
If we increase the space allocated to 1MemBF for storage to 1.5 times of the space used by ShBF$_\text{M}$, the FPR of 1MemBF is still a little more than that of ShBF$_\text{M}$ because hashing $k$ values into one or more words incurs serious unbalance in distributions of $1s$ and $0s$ in the memory, which in turn results in higher FPR.

\textcolor{greener}{\textit{Our results also show that the FPR of ShBF$_\text{M}$ is much smaller than that of 1MemBF when changing $k$ and $m$.}
Figure \ref{eva:Single:FP:k} and Figure \ref{eva:Single:FP:m} show the FPRs of ShBF$_\text{M}$ and 1MemBF for different values of $k$ and $m$, respectively.}

\presub
\subsubsection{ShBF$_{\text{M}}$ -- Memory Accesses}
\emph{Our results show that ShBF$_\text{M}$ answers a membership query using only about half the memory accesses and hash computations and twice as fast compared to BF.}
\textcolor{greener}{Our experiments for evaluating the number of memory accesses per query are similar to that for false positive rate, except that, now we query $2*n$ elements, in which $n$ elements belong to the set.
Figures \ref{eva:Single:Mem:n}, \ref{eva:Single:Mem:k}, and \ref{eva:Single:Mem:m} show the number of memory accesses for ShBF$_{M}$ and standard BF for different values of $n$, $k$, and $m$, respectively.
We also observed from our experiments that standard deviation in the results for ShBF$_{M}$ is also about half of that of standard BF.
}
%
%
%

\presub
\vspace{-0.05in}
\subsubsection{ShBF$_{\text{M}}$ -- Query Processing Speed}
\emph{Our results show that ShBF$_\text{M}$ has 1.8 and 1.4 times faster query processing speed compared to BF and 1MemBF, respectively}.
Although 1MemBF only needs one memory access per query, it needs $k+1$ hash functions. BFs are usually small enough to be stored in on-chip memory (such as caches, FPGA block RAM), thus the speed of hash computation will be slower than memory accesses. In contrast, our ShBF$_\text{M}$ reduces both hash computation and memory accesses. In our experiments, using those hashes which passed our randomness test, ShBF$_\text{M}$ exhibits faster query processing speed than that of 1MemBF. It is possible that 1MemBF is faster than ShBF$_\text{M}$ when using simple hash functions, but this probably incurs larger FPR.
\textcolor{greener}{Our experiments for evaluating the query processing speed are similar to that for memory accesses, except that, here we also compare with 1MemBF.
Figures \ref{eva:Single:Qps:n}, \ref{eva:Single:Qps:k}, and \ref{eva:Single:Qps:m} show the query processing speed for ShBF$_{M}$, standard BF, and 1MemBF for different values of $n$, $k$, and $m$, respectively.}

%
%

\presub
\vspace{-0.05in}
\subsection{ShBF$_{\text{A}}$ -- Evaluation}\postsub

In this section, we first validate the probability of a clear answer of ShBF$_{\text{M}}$ calculated in Table \ref{table:compare:ibf} using our experimental results.
Then we compare ShBF$_{\text{A}}$ with iBF in terms of FPR, memory accesses, and query processing speed.

\vspace{-0.05in}
\subsubsection{ShBF$_{\text{A}}$ -- Probability of Clear Answer}

\textit{Our results show that probability of clear answer for ShBF$_{\text{A}}$ calculated in Table \ref{table:compare:ibf} matches with the probability calculated experimentally.}
We performed experiments for both iBF and ShBF$\text{A}$ using two sets with 1 million elements such that their intersection had 0.25 million elements.
The querying elements hit the three parts with the same probability.
While varying the value of $k$, we also varied the value of $m$ to keep the filter at its optimal.
Note that in this case, iBF uses 1/7 times more memory than ShBF$\text{A}$.
We observe from Figure \ref{eva:twoset:FP:k} that the simulation results match the theoretical results, and the average relative error is 0.7\% and 0.004\% for iBF and ShBF$\text{A}$, respectively, which is negligible.
When the value of $k$ reaches 8, the probability of a clear answer reaches 66\% and 99\% for iBF and ShBF$\text{A}$, respectively.


\vspace{-0.05in}
\subsubsection{ShBF$_{\text{A}}$ -- Memory Accesses}
\emph{Our results show that the average number of memory accesses per query of ShBF$_\text{A}$ is 0.66 times of that of iBF.}
Figure \ref{eva:twoset:Mem:k} shows the number of memory accesses for different values of $k$.
We observed similar trends for different values of $m$ and $n$, but have not including the corresponding figures due to space limitation.

\presub
\vspace{-0.05in}
\subsubsection{ShBF$_{\text{A}}$ -- Query Processing Speed}
\emph{Our results show that the average query processing speed of ShBF$_\text{A}$ is 1.4 times faster than that of iBF.}
Figure \ref{eva:twoset:Qps:k} plots the the query processing speed of ShBF$_\text{A}$ and iBF for different values of $m$.

\vspace{-0.03in}
\subsection{ShBF$_{\times}$ -- Evaluation} \postsub

\vspace{-0.03in}

In this section, we first validate the correctness rate (CR) of ShBF$_{\times}$ calculated in Equation \eqref{equ:correct:ShBF_X}.
Then we compare ShBF$_{\times}$ with spectral BF \cite{spectralBF} \textcolor{bluer}{and CM Sketch \cite{CMsketch}}
in terms of CR, number of memory accesses, and query processing speed.
\textcolor{bluer}{The results for CM Sketch and Spectral BF are similar because their methods of recording the counts is similar.}

\presub
\vspace{-0.03in}
\subsubsection{ShBF$_{\times}$ -- Correctness Rate}
\vspace{-0.03in}
\emph{Our results show that the CR of ShBF$_{\times}$ calculated in Equation \eqref{equ:correct:ShBF_X} matches with the CR calculated experimentally.}
\emph{Our results also show that on average, the CR of ShBF$_{\times}$ is 1.6 times and \textcolor{bluer}{1.79 times of that of Spectral BF and CM Sketch, respectively.}}
For the experiments reported in this section, we set $c = 57$, $n=100,000$, and vary $k$ in the range $8\leqslant k\leqslant16$.
For spectral BF and CM sketch, we set use 6 bits for each counter.
\textcolor{greener}{For each value of $k$, as ShBF$_{\times}$ is more memory efficient, we use 1.5 times the optimal memory (\ie, $1.5*nk/ln2$) for all the three filters.}
Figure \ref{eva:ShBF_C:FP} shows the results from our experiments for CR.
We observe from this figure that the CR calculated through experiments matches with the CR calculated theoretically and the relative error is less than 0.08\%, which is negligibly small.
%
\presub
\vspace{-0.04in}
\subsubsection{ShBF$_{\times}$ -- Memory Accesses}
\vspace{-0.04in}
\emph{Our results show that the number of memory accesses of ShBF$_{\times}$ is smaller than that of spectral BF and CM Sketch for $k \geqslant 7$, and almost equal for $k < 7$.}
%
Figure \ref{eva:ShBF_C:MemAcc} plots the number of memory accesses of ShBF$_{\times}$, \textcolor{bluer}{CM Sketch}, and spectral BF, calculated from the same experiments that we used to plot Figure \ref{eva:ShBF_C:FP} except that $k$ ranges from 3 to 18.

\presub
\vspace{-0.03in}
\subsubsection{ShBF$_{\times}$ -- Query Processing Speed}
\vspace{-0.03in}
\emph{Our results show that ShBF$_{\times}$ is faster than spectral BF and \textcolor{bluer}{CM Sketch} when $k\geqslant11$.}
We evaluate the query processing speed of ShBF$_\times$, \textcolor{bluer}{CM Sketch,} and spectral BF using the same parameters as for Figure \ref{eva:ShBF_C:MemAcc}.
Figure \ref{eva:ShBF_C:Qps} plots the query processing speeds of ShBF$_\times$ and spectral BF.
We observe from this figure that when $k>11$, the average query processing speed of ShBF$_\times$ is over 3 Mqps.

%% file: conclusion.tex
\presec
\vspace{-0.03in}
\section{Conclusion} \postsec
\label{sec:conclusion}
\vspace{-0.03in}
The key contribution of this paper is in proposing Shifting Bloom Filter, a general framework to answer a variety of set queries.
We present how to use ShBF to answer three important set queries, \ie, membership, association, and multiplicity queries.
The key technical depth of this paper is in the analytical modeling of ShBF for each of the three types queries, calculating optimal system parameters, and finding the minimum FPRs.
We validated our analytical models through simulations using real world network traces.
Our theoretical analysis and experimental results show that ShBF significantly advances state-of-the-art solutions on all three types of set queries.